\documentclass[sigplan,10pt]{acmart}
\usepackage{array,makecell}
\usepackage{algorithm}
\usepackage{algpseudocode}
\usepackage{balance}
\usepackage{multirow}
\usepackage{diagbox}
\usepackage{amsmath}
\usepackage{graphics}
\usepackage{multicol}
\usepackage{enumitem}
\usepackage{mathtools}
\usepackage{epsfig}
\usepackage{xspace}
\usepackage{extarrows}
\usepackage{bm}
\usepackage{rotating}
\usepackage{listings}
\usepackage{eucal}
\usepackage[symbol]{footmisc}

\makeatletter
\providecommand{\bigsqcap}{
  \mathop{
    \mathpalette\@updown\bigsqcup
  }
}
\newcommand*{\@updown}[2]{
  \rotatebox[origin=c]{180}{$\m@th#1#2$}
}
\makeatother

\newcommand{\amir}[1]{\textcolor{black}{#1}}

\setcopyright{none}

\usepackage{booktabs}

\usepackage{subcaption}

\newtheorem{example}{Example}
\newtheorem{remark}{Remark}
\newtheorem{definition}{Definition}
\newtheorem{lemma}{Lemma}
\newtheorem{proposition}{Proposition}

\newcommand{\sle}{\sqsubseteq}
\newcommand{\sge}{\sqsupseteq}

\newcommand{\lfp}{\mathrm{lfp}}
\newcommand{\gfp}{\mathrm{gfp}}

\newcommand{\modelname}{{PTS}}
\newcommand{\modelnames}{\modelname\ }
\newcommand{\vars}{V}
\newcommand{\loct}{\loc_{\mathrm{t}}}
\newcommand{\locf}{\loc_{\mathrm{f}}}

\newcommand{\rdvar}{r}
\newcommand{\rdvars}{R}
\newcommand{\rdvarjdis}{\mathcal D}
\newcommand{\loc}{\ell}
\newcommand{\locs}{\mathit{L}}
\newcommand{\val}{\mathbf{v}}
\newcommand{\valrd}{{\mathbf{r}}}
\newcommand{\frmloc}{\ell^{\mathrm{src}}}
\newcommand{\toloc}{\ell^{\mathrm{dst}}}

\newcommand{\transcond}{\varphi}
\newcommand{\transprob}{p}
\newcommand{\transupdate}{\mbox{\sl upd}}

\newcommand{\fork}{F}
\newcommand{\sampset}{\mathcal{U}}

\newcommand{\trans}{{\tau}}
\newcommand{\transset}{\mathfrak T}
\newcommand{\pts}{\Pi}
\newcommand{\locin}{\loc_{\mathrm{init}}}
\newcommand{\valin}{\val_{\mathrm{init}}}
\newcommand{\invmap}{I}

\newcommand{\pstate}{\sigma}
\newcommand{\pstates}{\mathcal S}

\newcommand{\trace}{\Gamma}
\newcommand{\prob}{Quantitative Assertion Violation Analysis}
\newcommand{\probabbr}{\textit{QAVA}}
\newcommand{\uprobabbr}{\textit{UQAVA}}
\newcommand{\lprobabbr}{\textit{LQAVA}}

\newcommand{\cons}{\Xi}
\newcommand{\guard}{\Psi}
\newcommand{\condbody}{\Lambda}
\newcommand{\paracond}{{p,\alpha,\beta,\gamma}}
\newcommand{\precond}{\condbody_{\triangleleft}}
\newcommand{\postcond}{\condbody_{\triangleright}}
\newcommand{\relaxpostcond}{\overline{\condbody}^\paracond_{\triangleright}}
\newcommand{\cancond}{Con}

\newcommand{\epf}{\mathrm{vpf}}
\newcommand{\vpf}{\epf}
\newcommand{\pre}{\mathsf{ptf}}
\newcommand{\ptf}{\pre}
\newcommand{\usol}{\theta}
\newcommand{\lsol}{\theta}

\newcommand{\rterm}{\eta}
\newcommand{\eterm}[2]{\exp(\rterm(#1,#2))}

\newcommand{\ltrans}{\beta}
\newcommand{\dtrans}{\Delta}
\newcommand{\yican}[1]{{\color{red} #1}}
\newcommand{\hongfei}[1]{{\color{black} #1}}

\newcommand{\compalg}{\textsf{ExpLinSyn}}
\newcommand{\compalgs}{\compalg\ }
\newcommand{\heurisalg}{\textsc{HoeffdingSynthesis} }

\newcommand{\loweralg}{\textsf{ExpLowSyn}}
\newcommand{\loweralgs}{\loweralg\ }

\lstdefinelanguage{prog}
{
morekeywords={prob, if, then, else, fi, while, do, od, true, false, and, or, skip, switch, assert, exit, input},
sensitive = false
}

\algdef{SE}[DOWHILE]{Do}{DoWhile}{\algorithmicdo}[1]{\algorithmicwhile\ #1}

\newcommand{\jryremove}[1]{{}}

\renewcommand{\paragraph}[1]{\smallskip\noindent\textbf{\emph{#1.}}}

\begin{document}

\title{Quantitative Analysis of Assertion Violations in Probabilistic Programs}
\titlenote{Chinese authors are ordered by contribution, while Austrian authors are ordered alphabetically.}

 \author{Jinyi Wang}
 \authornote{Equal contribution}
 \affiliation{
   \institution{Shanghai Jiao Tong University}
 }
 \email{jinyi.wang@sjtu.edu.cn}
\author{Yican Sun}
\authornotemark[1]
 \affiliation{
   \institution{Peking University}
 }
 \email{sycpku@pku.edu.cn}
\author{Hongfei Fu}
\authornote{Corresponding author}
 \affiliation{
   \institution{Shanghai Jiao Tong University}
 }
 \email{fuhf@cs.sjtu.edu.cn}
\author{Krishnendu Chatterjee}
 \affiliation{
   \institution{IST Austria}
 }
 \email{krishnendu.chatterjee@ist.ac.at}
\author{Amir Kafshdar Goharshady}
 \affiliation{
   \institution{IST Austria}
 }
 \email{goharshady@gmail.com}

\begin{abstract}
We consider the fundamental problem of deriving quantitative bounds on the probability that a given assertion is violated in a probabilistic program. We provide automated algorithms that obtain both lower and upper bounds on the assertion violation probability. The main novelty of our approach is that we prove new and dedicated fixed-point theorems which serve as the theoretical basis of our algorithms and enable us to reason about assertion violation bounds in terms of pre and post fixed-point functions. To synthesize such fixed-points, we devise algorithms that utilize a wide range of mathematical tools, including repulsing ranking supermartingales, Hoeffding's lemma, Minkowski decompositions, Jensen's inequality, and convex optimization.

On the theoretical side, we provide (i)~the first automated algorithm for lower-bounds on assertion violation probabilities, (ii)~the first complete algorithm for upper-bounds of exponential form in affine programs, and (iii)~provably and significantly tighter upper-bounds than the previous approaches.
On the practical side, we show our algorithms can handle a wide variety of programs from the literature and synthesize bounds that are \amir{remarkably tighter than previous results, in some cases by thousands of orders of magnitude.}

\end{abstract}



\maketitle

\section{Introduction}

\smallskip
\noindent{\bf \em Probabilistic Programs.}
Extending classical imperative programs with the ability of sampling random values from predetermined probability distributions leads to probabilistic programs~\cite{gordon2014probabilistic}. Probabilistic programs are ubiquitous in virtually all parts of computer science, including machine learning~\cite{claret2013bayesian,scibior2015practical,roy2008stochastic}, robotics~\cite{thrun2002probabilistic}, and network analysis~\cite{netkat,netkat2,netkat3}. As a result, there are many probabilistic programming languages and their formal analysis is a central topic in programming languages and verification~\cite{AgrawalC018,SriramCAV,EGK12,pldi18,OLKMLICS2016,DBLP:journals/pacmpl/HarkKGK20,DBLP:conf/mfcs/KaminskiK15}.

\smallskip
\noindent{\bf \em Qualitative Analysis of Probabilistic Programs.}
The most well-studied problem in the qualitative analysis of probabilistic programs is that of \emph{termination}. Various notions of termination, such as finite-time termination~\cite{BG05,HolgerPOPL,DBLP:journals/toplas/ChatterjeeFNH18} and probability-1 (almost-sure) termination~\cite{SriramCAV,DBLP:conf/mfcs/KaminskiK15}, have been considered, and a wealth of methods have been proposed, e.g.~
patterns~\cite{EGK12}, abstraction~\cite{MM05}, martingale-based \cite{SriramCAV,DBLP:journals/toplas/ChatterjeeFNH18,ChatterjeeFG16,DBLP:conf/aplas/HuangFC18}, proof rules~\cite{OLKMLICS2016,DBLP:journals/jacm/KaminskiKMO18,mciver2017new},
and compositional~\cite{DBLP:journals/pacmpl/Huang0CG19}.

\smallskip
\noindent{\bf \em Quantitative Analysis of Probabilistic Programs.} Generally speaking, quantitative analyses of probabilistic programs are subtler and more complex than qualitative ones. Fundamental problems in this category include expected runtime analysis~\cite{DBLP:journals/jcss/BrazdilKKV15,DBLP:journals/jacm/KaminskiKMO18,ChatterjeeF19}, cost and resource analysis
~\cite{pldi18,ijcai18,DBLP:conf/pldi/Wang0GCQS19},
concentration bounds on the runtime
~\cite{DBLP:journals/toplas/ChatterjeeFNH18,DBLP:conf/tacas/KuraUH19,DBLP:journals/corr/abs-2001-10150},
and the focus of this work, namely assertion violation bounds.

\smallskip
\noindent{\bf \em Assertion Violation Bounds.} Consider a probabilistic program, together with one or more assertions at some points of the program. We are interested in the probability that an assertion is violated when assuming a given initial state for the program. Specifically, we focus on finding upper and lower bounds for the assertion violation probability. This problem was first
considered
in~\cite{SriramCAV} and has since become one of the most fundamental quantitative analyses in probabilistic programming. Previous methods include concentration inequalities~\cite{SriramCAV,DBLP:journals/toplas/ChatterjeeFNH18,ChatterjeeFG16,ChatterjeeNZ2017,DBLP:journals/corr/abs-2001-10150} and automata-based approaches~\cite{DBLP:journals/pacmpl/SmithHA19}. See Section~\ref{sec:rel} for a detailed comparison with previous works.

\smallskip
\noindent{\bf \em Our Contributions.} Our first theoretical contribution is that we provide novel fixed-point theorems which characterize assertion violation bounds in terms of pre and post fixed-point functions (Section~\ref{sec:theory}). We then focus on exponential bounds and use this characterization to obtain synthesis algorithms for both upper and lower bounds.
\hongfei{The reason why we choose exponential bounds is that they best suit the common situation where the assertion violation probability decreases drastically wrt a combination of program variables.} Our algorithmic contributions are as follows:
\begin{itemize}
\item In Section~\ref{sec:heuristicalg}, we provide a sound polynomial-time algorithm for upper-bound synthesis via repulsing supermartingales and Hoeffding's lemma. Our algorithm is guaranteed to obtain a significantly tighter bound than~\cite{ChatterjeeNZ2017}.
\item In Section~\ref{sec:soundcompalg}, we provide a sound and \emph{complete} synthesis algorithm for exponential upper-bounds (with linear combinations of program variables as the exponent) for affine programs. This is achieved by Minkowski decomposition and a dedicated quantifier elimination procedure.
\item In Section~\ref{sec:alglower}, we turn our focus to exponential lower-bounds and, through Jensen's inequality, obtain a sound polynomial-time algorithm for the synthesis of such bounds in almost-surely terminating affine programs.
\end{itemize} \smallskip
On the practical side, we provide extensive experimental results (Section~\ref{sec:evaluation}), showing that, over several classical programs, our bounds are significantly tighter than previous approaches, in some cases by \amir{\emph{thousands of orders of magnitude}}.

\smallskip
\noindent{\bf \em Novelty.} This work is novel in a number of directions: (a)~we provide dedicated fixed-point theorems for assertion violation analysis and use them as the theoretical basis of our algorithms, whereas previous methods relied on either concentration bounds or automata-based approaches; (b)~we provide automated algorithms for inferring both upper and lower bounds, whereas previous methods could only handle upper-bounds; (c)~each of our algorithms has non-trivial novel components such as our dedicated and efficient quantifier elimination method, or the application of Hoeffding's lemma and Jensen's inequality in the context of assertion violation analysis; (d)~our algorithm in Section~\ref{sec:soundcompalg} is \emph{complete} \hongfei{in the sense of finding a near-optimal template given any error bound}, while staying practical. This is the first such completeness result in assertion violation analysis of probabilistic programs.

\smallskip
\noindent{\bf \em Limitations.} Given that the problem we are attempting is undecidable in its most general case, our algorithms have the following limitations:
\hongfei{(a) they only focus on deriving exponential bounds over affine/polynomial programs;} (b)~our lower-bound results assume almost-sure termination, i.e.~they assume that the probabilistic program under analysis terminates with probability $1.$ While this is a routine assumption, our results depend on it; (c)~there is currently a trade-off between completeness and polynomial runtime. Our algorithm in Section~\ref{sec:soundcompalg} provides completeness but is not guaranteed to run in polynomial time. Conversely, our algorithms in Sections~\ref{sec:heuristicalg} and~\ref{sec:alglower} are polynomial-time but not necessarily complete. Nevertheless, they provide tighter bounds than previous methods (Remark~\ref{rem:stocinv}). Moreover, the trade-off is in theory, only. Our experimental results (Section~\ref{sec:evaluation}) show that \amir{our complete algorithm is extremely efficient in practice.}

\section{Preliminaries}
\label{sec:pts}

 Throughout this work, we use a Probabilistic Transition System (PTS) \cite{SriramCAV} to model and analyze each of our programs. A PTS is conceptually similar to a probabilistic control flow graph~\cite{DBLP:journals/toplas/ChatterjeeFNH18,ChatterjeeFG16}. Hence, translating an imperative probabilistic program into an equivalent PTS is a straightforward process.

\paragraph{Valuations} Let $X$ be a finite set of variables. A \emph{valuation} over $X$ is a function $\val: X \rightarrow \mathbb{R}.$ We denote the set of all valuations over $X$ by $\mathbb{R}^X.$ Moreover, we write $\val(x)$ to denote the value assigned by $\val$ to $x \in X$.

\paragraph{Program and Sampling Variables} In the sequel, we consider two disjoint sets of variables: (i) the set $\vars$ of \emph{program variables} whose values are determined by assignment statements in the program, and (ii)~the set $\rdvars$ of \emph{sampling variables} whose values are independently sampled from a predefined probability distribution each time they are accessed. For a sampling variable $r,$ we denote its distribution by $\rdvarjdis(r)$ and its support, i.e.~the set of all values that can be assigned to $r$, by $\sampset(r).$ We also define $\sampset = \prod_{r \in \rdvars} \sampset(r).$

\paragraph{Update Functions} An \emph{update function} $\transupdate$ is a function $\transupdate: \mathbb{R}^\vars \times \mathbb{R}^\rdvars \rightarrow \mathbb{R}^\vars$ that assigns a new valuation to program variables based on the current values of both program and sampling variables. Informally, we use update functions to model the effect of running a basic block of code.

We are now ready to define the notion of a PTS. We extend the definition in~\cite{SriramCAV} with assertion violations.

\paragraph{Probabilistic Transition Systems}
A \emph{Probabilistic Transition System} is a tuple $\pts = (\vars,\rdvars,\rdvarjdis,\locs,\transset, \locin, \valin, \loct, \locf),$ where:
\begin{itemize}
\item
$\vars$ is a finite set of \emph{program variables}.
\item
$\rdvars$ is a finite set of \emph{sampling variables} and $\rdvars \cap \vars = \emptyset$.
\item $\rdvarjdis$ is a function that assigns a probability distribution $\rdvarjdis(r)$ to each sampling variable $r \in \rdvars$.
\item $\locs$ is a finite set of \emph{locations} or \emph{program counters}.
\item $\locin \in \locs$ is the initial location and $\valin \in \mathbb{R}^\vars$ is the initial valuation for program variables.
\item $\locf, \loct\in \locs.$
Intuitively,
$\loct$ represents program termination and $\locf$ corresponds to assertion violation.
\item
$\transset$ is a finite set of \emph{transitions}. Each transition $\tau \in \transset$ is a tuple $\tau = \langle \frmloc, \transcond, \fork_{1}, \fork_{2}, \cdots, \fork_{k}\rangle$ such that
		\begin{itemize}[leftmargin=.5em,label={\tiny$\blacksquare$}]
			\item $\frmloc \in \locs \setminus \{\loct, \locf\}$ is the \emph{source} location;
                \item $\transcond$ is a logical formula over valuations on $\vars$ which serves as the transition's \emph{guard} or \emph{condition};
			\item Each $\fork_{j}$ is called a \emph{fork} and is of the form $\fork_j=\langle \toloc_{j}, \transprob_{j}, \transupdate_{j}\rangle$ in which
			     $\toloc_{j}\in \locs$ is the \emph{destination} location,
			     $\transprob_{j} \in (0, 1]$ is the probability assigned to this fork, and
			     $\transupdate_{j}$ is an update function.
			     It is guaranteed that $\sum_{j=1}^{k}\transprob_j = 1$.
		\end{itemize}
\end{itemize}
A \emph{state} of $\pts$ is a pair $\pstate = (\loc, \val) \in \locs \times \mathbb{R}^\vars$ that consists of a location and a valuation. In the sequel, we assume that we have fixed a PTS $\pts = (\vars,\rdvars,\rdvarjdis,\locs,\transset, \locin, \valin, \loct, \locf)$.

\paragraph{Intuitive Description}
The program starts at $(\locin, \valin).$
A transition $\langle \frmloc, \transcond, \fork_{1}, \fork_{2}, \cdots, \fork_{k}\rangle$ with $\fork_{j}=\langle \toloc_{j}, \transprob_{j}, \transupdate_{j}\rangle$
states that if the current location is $\frmloc$ and the current valuation $\val$ of program variables satisfies the condition $\transcond$, then
each fork $\fork_{j}$ is chosen and applied with probability $\transprob_{j}.$ When we apply $\fork_j,$ the next location is $\toloc_{j}$ and the next valuation is $\transupdate_{j}(\val,\valrd)$,
in which $\valrd \in \mathbb R^\rdvars$ is obtained by independently sampling a value for each $r \in \rdvars$ according to $\rdvarjdis(r).$

\begin{example}
	Figure \ref{fig:pts21} shows a PTS representation of a program. Oval nodes represent locations and square nodes model the forking behavior of transitions. An edge entering a square node is labeled with the condition of its respective transition. The numbers in green denote the probability of each fork, while blue expressions show the update functions.
\end{example}

\paragraph{Additional Assumption} To disallow non-determinism and undefined behavior, we require that:
(i) any two transitions $\tau \neq \tau'$ with the same source location be mutually exclusive, i.e.~ if their guards are $\transcond$ and $\transcond'$, then $\transcond \land \transcond'$ is unsatisfiable;
(ii) the set of transitions
be \emph{complete}, i.e.
for every location $\loc$ other than $\loct,\locf$ and every valuation $\val \in \mathbb R^\vars$, there must exist a transition out of $\loc$
whose guard condition is satisfied by $\val$.

\paragraph{Semantics}
The semantics of $\pts$ is formalized by its corresponding \emph{PTS process} $\Gamma.$
$\Gamma$ is a stochastic process $\{\hat{\sigma}_n\}_{n\ge 0}$ on states.
Given the current state $\hat{\sigma}_n=(\hat{\loc}_n,\hat{\val}_n)$, if $\hat{\loc}_n \not\in \{\locf, \loct\}$, the transition is specified as follows:
(1)~Take the unique transition $\langle \hat{\loc}_n,\transcond,\fork_{1},\dots, \fork_{k}\rangle$ with $\hat{\val}_n\models\transcond$.
(2)~Choose the fork $\fork_{j} = \langle \toloc_j, p_j, \transupdate_{j} \rangle$ with probability $p_j$.
(3)~Obtain a valuation $\valrd$ over our sampling variables $\rdvars$ by sampling each $r \in \rdvars$ independently according to $\rdvarjdis(r).$
(4) Apply the chosen fork: $(\hat{\loc}_{n+1}, \hat{\val}_{n+1})=(\toloc_j, \transupdate_j(\hat{\val}_n,\valrd))$.
If either $\loct$ or $\locf$ is reached, $(\hat{\loc}_{n+1}, \hat{\val}_{n+1}) = (\hat{\loc}_n, \hat{\val}_n)$. See Appendix \ref{appendix:preliminary} for details.

\begin{figure}[htbp]
		\lstset{language=prog}
	\lstset{tabsize=3}
	\newsavebox{\progrunningexampleb}
	\begin{lstlisting}[mathescape]
			$x$:=$40$; $y$:=0;
			while($x\leq 99\land y\leq 99$):
				if prob($0.5$):
					$\langle x,y\rangle:=\langle x+1,y+2\rangle$
				else:
					$\langle x,y\rangle:=\langle x+1,y\rangle$
			assert($x \ge 100$)
	\end{lstlisting}
	\includegraphics[width=0.4\textwidth]{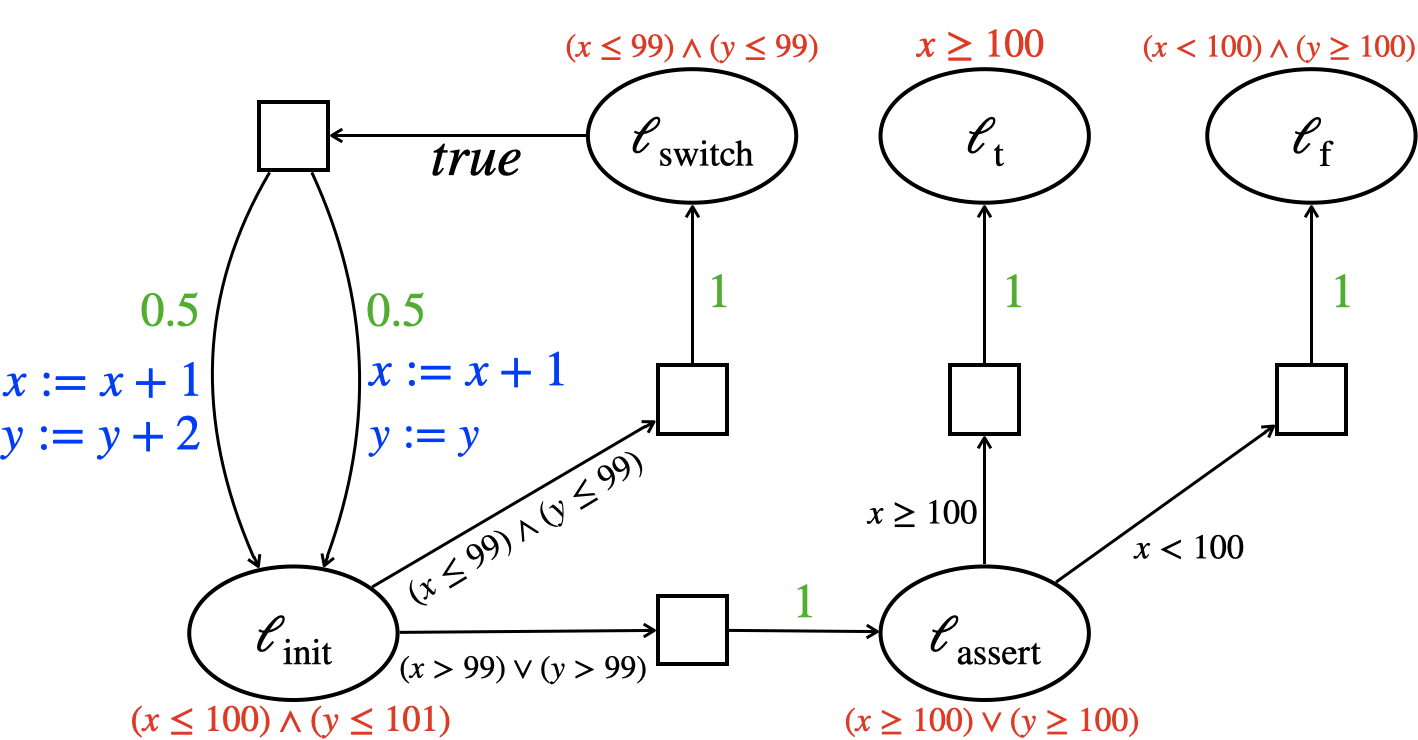}
	\caption{A Probabilistic Program (top) and its PTS (bottom)}
	\label{fig:pts21}
\end{figure}

\paragraph{Paths and Reachability} Let $\pts$ be a PTS. A \emph{path} in $\pts$ is an infinite sequence of states $\pstate_0, \pstate_1,\dots$ such that:
\begin{itemize}
\item $\pstate_0=(\locin, \valin)$, and
\item for each $n\ge 0$, the states $\pstate_n=(\loc_n, \val_n)$ and $\pstate_{n+1}=(\loc_{n+1}, \val_{n+1})$ satisfy one of the following cases:
\begin{itemize}
\item ${\loc}_n\not\in \{\loct, \locf\}$ and there is a transition $\langle {\loc}_n,\transcond,\fork_{1},\\\dots, \fork_{k}\rangle$ with a fork $F_j$
such that $\val\models\transcond$, $\fork_{j}=\langle \loc_{n+1}, \transprob_{j}, \transupdate_{j}\rangle$, and
${\val}_{n+1}=\transupdate_{j}(\val_{n}, \valrd)$ for some $\valrd\in \sampset$;
\item $\loc_n\in \{\loct, \locf\}$ and $(\loc_{n+1}, \val_{n+1})=(\loc_n, \val_n)$.
\end{itemize}
\end{itemize}
A state $\pstate$ is \emph{reachable} if there is a path $\pstate_0, \pstate_1,\dots$ such that $\pstate_n=\pstate$ for some $n\ge0$.
$\pstates$ denotes the set of reachable states.

\paragraph{Invariants}
For a PTS $\pts,$ an \emph{invariant} is a function $\invmap: \locs \rightarrow 2^{\mathbb{R}^\vars}$ that assigns to each location $\loc \in \locs,$ a subset $\invmap(\loc)$ of valuations over program variables such that $\pstates\subseteq \{(\loc, \val)\mid \val \in \invmap(\loc)\}$. An invariant over-approximates reachable states.

\begin{example}
	In Figure~\ref{fig:pts21}, every location has an associated formula in red, representing an invariant at that location.
\end{example}

We now formalize the central problem studied in this work:

\paragraph{\prob{}}
Given a PTS $\pts$ and an invariant $\invmap$, the {\prob{}} (\probabbr) problem is to infer upper and lower bounds for the probability that the PTS process $\trace,$ starting at initial state $(\locin, \valin),$ ends with assertion violation. Formally, the goal is to synthesize upper and lower bounds for
$$\Pr[\exists n.~\hat{\loc}_n = \locf\ |\ \hat{\sigma}_0 = (\locin,\valin)].$$
We abbreviate the upper-bound analysis problem as \uprobabbr, and its lower-bound counterpart as \lprobabbr.

\section{Illustration and Motivating Examples}

In this section, we illustrate our approach over three examples from different application domains. We will provide a more formal treatment in Section~\ref{sec:theory}.
\begin{itemize}
\item In Section~\ref{sec:rabbit}, we show the basic ideas using an example program, taken from the literature on stochastic invariants, that models a tortoise-hare race.
\item In Section~\ref{sec:rdwalk}, we show how our approach can be used to solve one of the most classical problems in probabilistic termination analysis, namely concentration bounds.
\item While the two examples above rely on upper-bounds, in Section~\ref{sec:hardware}, we illustrate our approach for lower-bounds and apply it to quantify the probability of error in computations performed on unreliable hardware.
\end{itemize}

\subsection{Tortoise-Hare Race}
\label{sec:rabbit}

	Consider the program of Figure~\ref{fig:pts21}, which is often encountered in works on stochastic invariants (e.g.~\cite{ChatterjeeNZ2017}). This program models a classical tortoise-hare race. The variable $x$ represents our tortoise's position, while $y$ is the position of the hare. Initially, the tortoise is given a $40$-unit edge.
	In each iteration, the tortoise moves one step forward, and the hare either takes a two-unit jump or rests. The finish line is at position $100$. So, the \textbf{assert} corresponds to a win for the tortoise.
	We aim to obtain an upper-bound for the hare's winning probability, i.e.~the probability of assertion violation.

We establish such an upper-bound by relying on fixed-point theorems. The fundamental idea is to synthesize a function $f(x^*, y^*)$ that serves as an overapproximation of the probability that the assertion is eventually violated, assuming that we start running the program from the entry point of our while loop with variable valuations $x=x^*$ and $y=y^*.$ We can set up the following constraints over $f(x^*, y^*)$:
\begin{enumerate}
	\item[(I)] $\forall x^*, y^*. \quad f(x^*, y^*)\ge 0$;
	\item[(II)] $\forall x^*,y^*. \quad x^*\le 99~\land~ y^*\ge 100 \Rightarrow  f(x^*,y^*)\ge 1$;
	\item[(III)] $\forall x^*,y^*. \quad x^*\le 99~\land~ y^*\le 99 \Rightarrow f(x^*,y^*)\ge 0.5 \cdot f(x^*+1,y^*+2)+ 0.5 \cdot f(x^*+1,y^*)$.
\end{enumerate} \smallskip
Informally, constraint (I) is natural since probability values are always non-negative.
Constraint (II) requires that when the program terminates with an assertion violation, the value of $f$ is at least $1$.
Finally, constraint (III) is applied when another iteration of the loop is about to run and enforces that our approximation of the probability of violating the assertion at this point is no less than its expectation after the execution of one iteration. More formally, this condition is
derived from the fixed-point theorem and states that $f(x^*,y^*)$ is a prefixed-point (See Theorem~\ref{thm:epf}).

By Tarski's fixed-point theorem (Theorem~\ref{thm:tarski}), any function $f(x^*, y^*)$ that satisfies the constraints (I)--(III) serves as an upper-bound for the assertion violation probability given any initial valuation $x^*, y^*$. Specifically, we focus on synthesizing such a function $f$ so that we can use $f(40, 0)$ as our upper-bound on the probability of assertion violation.

In this work, we focus on assertions whose probability of violation decreases exponentially with respect to a combination of program variables.
We follow a template-based method (see e.g.~\cite{SriramCAV,DBLP:journals/toplas/ChatterjeeFNH18,ChatterjeeFG16}) and set up an exponential template $f(x^*,y^*):=\exp(a\cdot x^* + b\cdot y^* + c)$.
Our goal is to synthesize values for the variables $a, b, c$ such that $f$ satisfies constraints (I)--(III) above, while simultaneously minimizing $f(40, 0)$.
This template specifies that the assertion violation probability decreases exponentially with respect to the linear expression $a\cdot x^* + b\cdot y^* + c$.
Thus, it suffices to solve the following optimization problem with unknown variables $a, b, c$:
$$\textstyle \textbf{Minimize}~~~~~ \exp(40\cdot a + 0\cdot b + c)$$
$$
\textstyle \textbf{Subject to}~~~~~\text{constraints (I)--(III)}
$$

In general, solving such optimization problems is hard, since the constraints are universally quantified and involve exponential terms.
Surprisingly, in Section~\ref{sec:soundcompalg}, we show that a large class of optimization problems of this kind, including the problem above, can be exactly solved through convex programming.
By solving this optimization problem, we derive $a\approx-1.19,b\approx 4.26,c\approx 31.79$, and the optimal value is $\approx \exp(-15.697)\approx 1.524\cdot 10^{-7}$. Hence, the probability that the assertion is violated is at most $1.524\cdot 10^{-7}.$

\subsection{Concentration Bounds}
\label{sec:rdwalk}

Concentration analysis of termination time is a fundamental problem in probabilistic programming~\cite{DBLP:journals/toplas/ChatterjeeFNH18} whose goal is to derive rapidly-decreasing upper-bounds in terms of $n$ for the probability that a probabilistic program does not terminate (continues running) after $n$ steps.
To model
    this problem in our framework,
we introduce a new program variable $t$ that keeps track of the running time and is incremented in every iteration/step of the program.
We also add the assertion $\textbf{assert}(t<n)$ at the endpoint of the program. Here, $n$ is either a user-specified natural number or a fresh variable.

\begin{figure}
	\lstset{language=prog}
	\lstset{tabsize=3}
	\begin{lstlisting}[mathescape]
	$x$:=$0$; $t$:=$0$;
	while($x\leq 99$):
		switch:
			prob($0.75$):$\langle x,t\rangle$:=$\langle x+1,t+1\rangle$
			prob($0.25$):$\langle x,t\rangle$:=$\langle x-1,t+1\rangle$
      assert($t\le 500$)
	\end{lstlisting}
		\includegraphics[width=0.4\textwidth]{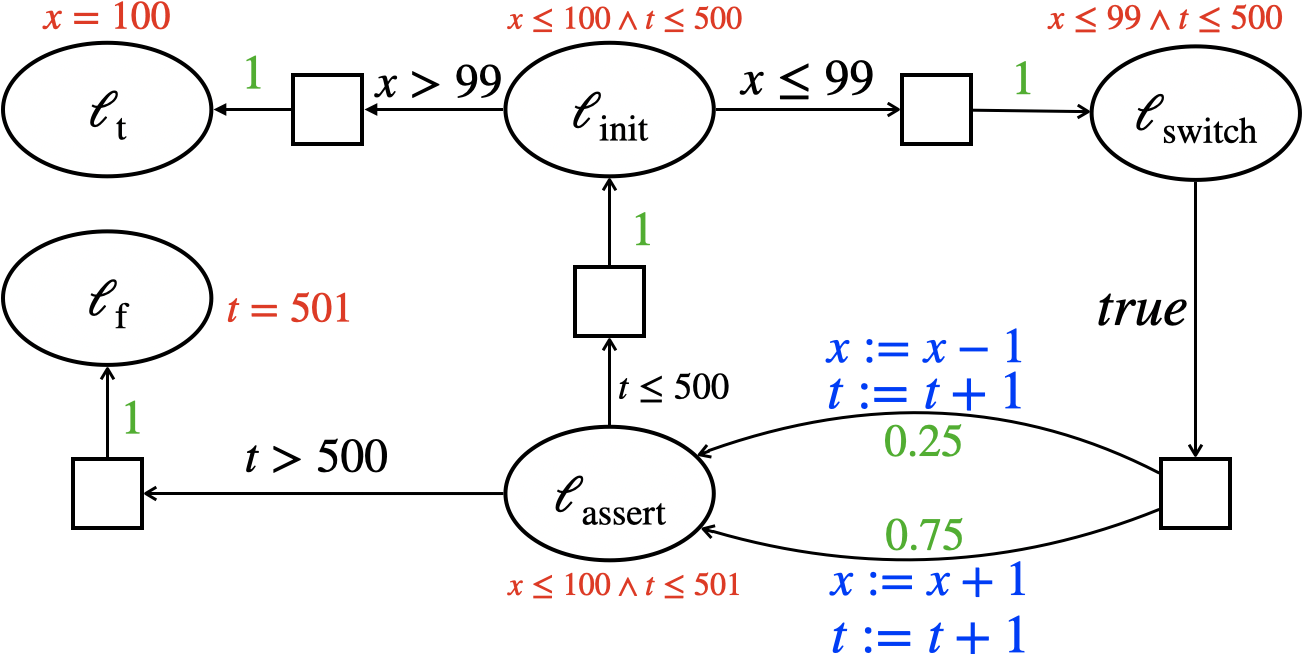}
		\caption{An Asymmetric Random Walk}
		\label{fig:pts22}
	\end{figure}

	As an example, consider the program in Figure~\ref{fig:pts22} which models an asymmetric random walk.
	In this program, the variable $x$ represents our current position in the random walk. The \textbf{switch} statement within the loop body specifies that at each step, we move forwards with probability $\frac34$ and backwards with probability $\frac14.$
The variable $t$ records the number of loop iterations. At the end of the loop body, the \textbf{assert} statement checks whether the program has run for at most $500$ iterations.
Our goal is to find an upper-bound for the probability of violation of this assertion, i.e.~the probability that our asymmetric random walk does not terminate (does not reach $x > 99$) within $500$ steps.

Again, we synthesize a function $f(x^*,t^*)$ that over-estimates the assertion-violation probability assuming the program is started at the while loop with valuation $x=x^*, t=t^*$. Such a function should satisfy the following constraints:
\begin{enumerate}
	\item[(I)] $\forall x^*,t^*. \quad f(x^*,t^*)\ge 0$;
	\item[(II)] $\forall x^*,t^*. \quad x^*\le 100\land t^*\ge 501 \Rightarrow f(x^*,t^*)\ge 1$;

	\item[(III)] $\forall x^*,t^*. \quad x^*\le 99\land t^*\le 500 \Rightarrow 0.25\cdot f(x^*-1,t^*+1) + 0.75\cdot f(x^*+1,t^*+1) \le f(x^*,t^*)$.
\end{enumerate}

The intuition behind these constraints is similar to the previous section. By Tarski's fixed-point theorem (Theorem~\ref{thm:tarski}), any function $f$ satisfying these constraints is an upper-bound on the assertion-violation probability. Given the initial values $x=t=0$, by letting $f(x^*,t^*):=\exp(a\cdot x^* + b\cdot t^* + c)$ and solving for $a, b, c$ (See Section~\ref{sec:soundcompalg} for details), we obtain $a\approx -0.351, b\approx 0.124, c\approx -27.181$. The assertion violation probability is at most $f(0, 0) \approx \exp(-27.181)\approx 1.569\cdot 10^{-12}$.

\subsection{Computing on Unreliable Hardware}
\label{sec:hardware}

Consider an unreliable hardware that might malfunction with a tiny probability at each execution step and cause the program to collapse or compute erroneously.
Reliability analysis of programs run over unreliable hardware is an active area of research (see e.g.~\cite{DBLP:conf/oopsla/CarbinMR13,DBLP:journals/pacmpl/SmithHA19}).
We now show how the reliability analysis can be reduced to the derivation of \emph{lower-bounds} for the probability of assertion violation, and provide an outline of our approach for deriving such lower-bounds.

Take the random walk example from the previous section, but assume that it is run on an unreliable hardware and any iteration may fail with probability $p=10^{-7}$.
Our goal is to derive a lower-bound for the probability that the random walk executes correctly until termination.
By incorporating hardware failure into the random walk,
we get the program in Figure~\ref{fig:pts23}. The only difference with the original random walk is that in each loop iteration, the hardware fails with probability $p.$ This is modeled by the \textbf{exit} statement.
We delibrately have the assertion $\textbf{false}$ at the end of the program so that the assertion fails iff there is \textit{no hardware failure} during the whole execution.
Thus, we are aiming to synthesize a lower-bound for the probability of assertion violation.
	\begin{figure}
	\lstset{language=prog}
	\lstset{tabsize=3}
	\begin{lstlisting}[mathescape]
	$x$:=$1$;
	while($x\leq 99$):
		switch:
			prob($p$):exit
			prob($0.75\cdot (1-p)$):$x$:=$x+1$
			prob($0.25\cdot (1-p)$):$x$:=$x-1$
	assert($\textbf{false}$)
	\end{lstlisting}
		\includegraphics[width=0.4\textwidth]{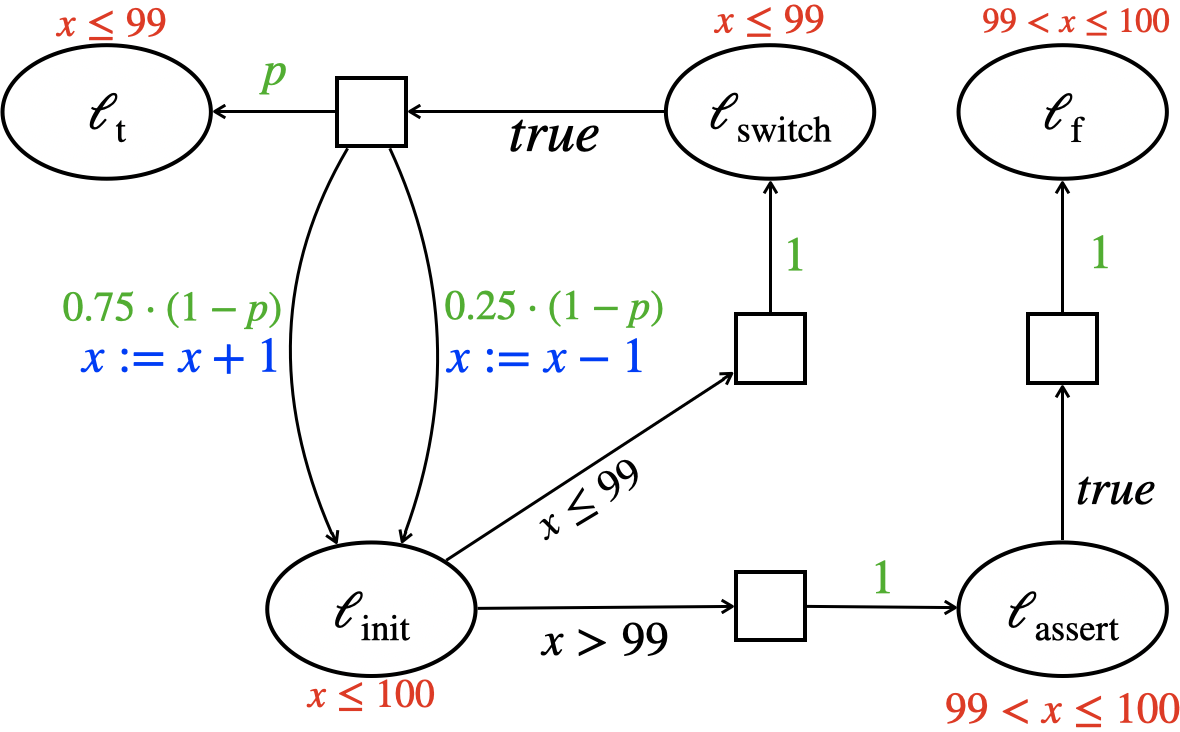}
		\caption{Random Walk Run with Unreliable Hardware}
		\label{fig:pts23}
	\end{figure}

Since we need to infer a lower-bound instead of an upper-bound, we will synthesize a function $f(x^*)$ at the entry point of the loop that always under-estimates the probability of assertion violation.
We establish a new fixed point theorem (Theorem \ref{thm:uniq}) by which the function $f$ should satisfy:
\begin{enumerate}
	\item[(I)] $\forall x^*. \quad x^*\le 100 \Rightarrow 0\le f(x^*)\le 1$;

	\item[(II)] $\forall x^*. \quad x^*\le 99 \Rightarrow f(x^*)\leq 0.75 \cdot (1-p)\cdot  f(x^*+1) + 0.25\cdot (1-p) \cdot f(x^*-1) + p\cdot 0$
\end{enumerate}

These constraints are, in a sense, duals of the constraints used for upper-bounds. The differences are that, in constraint (I), we restrict the value to be \emph{at most} $1$ and that, in constraint (II), we have a post fixed-point rather than a pre fixed-point, i.e.~$f(x^*)$ is \emph{less than or equal to} its expectation after the execution of one iteration.

As in the previous cases, we set up an exponential template $f(x^*):=\exp(a\cdot x^*  + b)$. Note that the initial value of $x$ is $1.$ Therefore, in order to obtain as tight a lower-bound as possible, we need to maximize $f(1)$.
So, we have to solve the following optimization problem with unknown variables $a, b$:
$$\textstyle \textbf{Maximize}~~~~ \exp(a + b)$$
$$
\textstyle \textbf{Subject to}~~~~~\text{constraints (I)--(II)}
$$
As we will see in Section~\ref{sec:alglower}, (I) can be transformed into an equivalent collection of linear constraints over $a, b$ using Farkas' Lemma.
In (II), we divide both sides by $\exp(a\cdot x^* + b)$: $$\textstyle 0.75\cdot (1-p)\cdot \exp(a) + 0.25\cdot (1-p)\cdot \exp(-a) \ge 1.$$
This is not a convex inequality.
Hence, we cannot apply convex programming to solve this optimization problem.
Instead, we use Jensen's inequality (Theorem~\ref{thm:jensen}) to relax these exponential constraints to linear ones. Concretely, (II) is satisfied if:
	$0.75\cdot a + 0.25\cdot(-a) \ge -\ln\left(1-p\right)$.
After these transformations, since maximizing $\exp(a+b)$ is equivalent to maximizing $a+b$,
the problem is relaxed and reduced to a linear programming instance, which can be solved efficiently.
In this case, we obtain $a\approx 2\cdot 10^{-7},b\approx -2\cdot 10^{-5}$. So, our lower-bound is $\approx \exp(-1.98\cdot 10^{-5}) \approx 0.99998$.

\section{A Fixed-Point Approach to QAVA}
\label{sec:theory}

In this section, we show how fixed-point theorems can be applied to the QAVA problem. Our results provide a theoretical basis for obtaining upper and lower bounds on the assertion violation probability.

\subsection{Lattices and Fixed-point Theorems}

\paragraph{Suprema and Infima} Given a partial order $\sle$ over a set $K,$ and a subset $K' \subseteq K,$ an \emph{upper-bound} of $K'$ is an element $u \in K$ that is larger than every element of $K',$ i.e.~$\forall k' \in K'.~k' \sle u.$ Similarly, a \emph{lower-bound} for $K'$ is an element $l$ that is smaller than every element of $K',$ i.e.~$\forall k' \in K'.~l \sle k'.$ The \emph{suprema} of $K',$ denoted by $\bigsqcup K'$, is an element $u^* \in K$ such that $u^*$ is an upper-bound of $K'$ and for every upper-bound $u$ of $K',$ we have $u^* \sle u.$ Similarly, the \emph{infima} $\bigsqcap K'$ is a lower-bound $l^*$ of $K'$ such that for every lower-bound $l$ of $K',$ we have $l \sle l^*.$ We also define $\bot\!:=\!\bigsqcap K$ and $\top\!:=\!\bigsqcup K.$ In general, suprema and infima may not exist.

\paragraph{Complete Lattice} A partial order $(K, \sle)$ is called a \emph{complete lattice} if every subset $K'\subseteq K$ has an suprema and a infima.

\paragraph{Monotone Functions} Given a partial order $(K, \sle)$, a function $f: K \to K$ is called \textit{monotone} if for every $k_1 \sle k_2$ in $K$, we have $f(k_1) \sle f(k_2).$

\paragraph{Continuity} Given a complete lattice $(K, \sle),$ a function $f: K \to K$ is called \emph{continuous} if for every increasing chain $k_0 \sle k_1 \sle \ldots$ in $K,$ we have $f(\bigsqcup \{k_n\}_{n=0}^\infty) = \bigsqcup \{f(k_n)\}_{n=0}^\infty,$ and \emph{cocontinuous} if for every decreasing chain $k_0 \sge k_1 \sge \ldots$ of elements of $K,$ we have $f(\bigsqcap \{k_n\}_{n=0}^\infty) = \bigsqcap \{f(k_n)\}_{n=0}^\infty.$

\paragraph{Fixed-Points} Given a complete lattice $(K, \sle)$ and a function $f: K \to K,$ an element $x \in K$ is called a \emph{fixed-point} if $f(x) = x.$ Moreover, $x$ is a \emph{pre fixed-point} if $f(x) \sle x$ and a \emph{post fixed-point} if $f(x) \sge x.$ The \emph{least fixed-point} of $f$, denoted by $\lfp f,$ is a fixed-point that is smaller than every fixed-point under $\sle.$ Analogously, the \emph{greatest fixed-point} of $f,$ $\gfp f,$ is a fixed-point that is larger than all fixed-points.

\begin{theorem}[\textit{Knaster{-}Tarski}~\cite{KnasterTarski}]
\label{thm:tarski}

 Let $(K, \sle)$ be a complete lattice and $f:K \to K$ a monotone function. Then, both $\lfp\ f$ and $\gfp\ f$ exist. Moreover, $\lfp\ f$ is the infima of all pre fixed-points, and $\gfp\ f$ is the suprema of all post fixed-points.
\begin{align}
	\textstyle \lfp\ f & \textstyle = \mathop{\bigsqcap} \left\{x\ |\ f(x)\sle x\right\}\\
	\textstyle \gfp\ f & \textstyle = \mathop{\bigsqcup} \left\{x\ |\ x\sle f(x)\right\} \label{formula:upper}
\end{align}
\end{theorem}

\noindent The next theorem provides a construction for the fixed-points.
\begin{theorem}[\cite{Sangiorgibook}]
\label{thm:kleene}

Let $(K, \sle)$ be a complete lattice and $f: K \to K$ be an continuous function. Then, we have
	$$\textstyle \lfp\ f = {\textstyle \mathop{\bigsqcup}_{i \ge 0}} \left\{f^{(i)}(\bot)\right\}.$$
Analogously, if $f$ is cocontinuous, we have
	$$\textstyle \gfp\ f = {\textstyle \mathop{\bigsqcap}_{i \ge 0}} \left\{f^{(i)}(\top)\right\}.$$

\end{theorem}

\subsection{Fixed-point Theorems in QAVA}

\paragraph{Violation Probability Function} We start with the \emph{violation probability function} $\vpf$. Intuitively, $\vpf(\loc, \val)$ is the probability that the PTS process $\Gamma$ starting from the state $(\loc, \val)$ ends with an assertion violation. Formally,
$$\textstyle \epf(\loc, \val):=\Pr\left[\exists n.~ \hat{\loc}_n = \locf\ |\ \hat{\sigma}_0 = (\loc,\val)\right].$$

\paragraph{Sketch of the Method} Our goal is to derive upper and lower bounds on $\epf(\locin, \valin)$. We define a set $\mathcal K^M$ of \emph{state functions} equipped with a natural partial order, as well as a \emph{probability transformer function} $\ptf: \mathcal K^M \to \mathcal K^M.$ We then apply Theorem \ref{thm:kleene}
to show that $\vpf$ is the least fixed-point of $\ptf.$

\paragraph{State Functions} Let $M \in [1, \infty).$ We define $\mathcal K^M$ as the set of all functions $f: \pstates \to [0, M]$ that map each reachable state to a real number between $0$ and $M$ and satisfy the following:
\begin{itemize}
	\item $\forall \val \in \mathbb R^\vars.~ f(\loct, \val) = 0,$ and
	\item $\forall \val \in \mathbb R^\vars.~ f(\locf, \val) = 1.$
\end{itemize}
 The partial order $\sle$ on $\mathcal K^M$ is defined standardly, i.e.~for every $f,f'\in \mathcal K^{M}$, we have $f\sle f'$ iff $\forall \pstate \in \pstates.~ f(\pstate) \leq f'(\pstate)$.
 It is straightforward to verify that $(\mathcal K^M, \sle)$ is a complete lattice. Furthermore, its smallest (resp.~greatest) element $\bot^{M}$ (resp.~$\top^M$) is the function whose value is $0$ (resp.~$M$) at all states $(\loc,\val)\in \pstates$ where $\loc \notin \{\loct,\locf\}$.
 We could similarly define $\mathcal K^{\infty}$ as the set of all functions $f:\pstates\to \mathbb [0,\infty)\cup \{\infty\}$, and the complete lattice structure on $\mathcal K^{\infty}$.
 See Appendix \ref{appendix:theory} for details.

\paragraph{Probability Transformer Function} The \emph{probability transformer function} $\ptf^M : \mathcal K^M \to \mathcal K^M$ is a higher-order function that computes the expected value of a given function after one step of PTS execution. Formally, it is defined as follows:
\begin{itemize}
	\item $\ptf^M(f)(\locf, \val) := 1;$
	\item $\ptf^M(f)(\loct, \val) := 0;$
	\item If $\loc \not\in \{\locf, \loct\},$ then for every valuation $\val \in \mathbb R^\vars,$ there exists a unique transition $\tau = (\loc, \varphi, F_1, \ldots, F_k)$ with $\val \models \varphi.$ Let its $i$-th fork be $F_i = (\toloc_i, p_i, \transupdate_i).$ Then,
	$$
	\textstyle \ptf^M(f)(\loc, \val) := \sum_{i=1}^k p_i \cdot \mathbb E \left[ f(\toloc_i, \transupdate_i(\val, \valrd)) \right]
	$$
	where $\valrd \in \mathbb R^\rdvars$ is sampled according to $\rdvarjdis.$
\end{itemize}

We can now obtain our first fixed-point theorem for QAVA.

\begin{theorem}[Proof in Appendix~\ref{appendix:theory}] \label{thm:epf}
	$\lfp\ \pre^{\infty} = \epf$.
\end{theorem}
\begin{proof}[Proof Sketch]
We apply Theorem \ref{thm:kleene} by plugging in $\bot^{\infty}$. The result follows by expanding the function composition.
\end{proof}

\paragraph{Upper Bounds on $\vpf$} By combining the theorem above with Kanster-Tarski's fixed-point theorem, we see that to infer an upper-bound for $\vpf,$ it suffices to find a pre fixed-point $\theta \in \mathcal K^\infty$ (as illustrated in Section~\ref{sec:rabbit}).

\paragraph{Lower Bounds on $\vpf$} Theorem \ref{thm:tarski} only provides lower bounds on the greatest fixed-point, but not the least one. Nevertheless, we can use these bounds if we can guarantee the uniqueness of our fixed-point, then every post fixed-point $\theta \in \mathcal K^M,$ for some $M \geq 1,$ would be a lower-bound on $\vpf.$

\paragraph{Almost-sure Termination} To achieve the desired uniqueness property, we further assume that our PTS terminates almost-surely \hongfei{(for all $\pstate\in\pstates$)}. Formally, $\Pr[\exists n.~ \hat{\loc}_n = \loct \lor \hat{\loc}_n = \locf~\hongfei{\mid (\locin, \valin)=(\loc, \val)}] = 1$ for all $(\loc,\val)\in \pstates$.
We show that under this assumption, the fixed-point is unique.

\begin{theorem} \label{thm:uniq}
	Let $\pts$ be an almost-surely terminating PTS. Then for every $M \geq 1$, we have $\epf = \lfp\ \pre^{M} = \gfp\ \pre^{M}.$
\end{theorem}

\begin{proof}[Proof Sketch]
Since $M$ is finite, $\pre^M$ is both continuous and cocontinuous. By plugging in the concrete form of $\bot^M$ and $\top^M$ into the formula of Theorem \ref{thm:kleene}, and using the definition of almost-sure termination assumption, one can derive the desired result. See Appendix~\ref{appendix:theory} for details.
\end{proof}

\hongfei{
\begin{remark}
Given the almost-sure termination assumption, one may argue that the lower-bound problem can be transformed into the upper-bound problem by swapping $\loct$ and $\locf$, as a lower bound for assertion violation is an upper bound for no assertion violation.
However, through Theorem~\ref{thm:uniq} we reduce the lower-bound problem to post fixed-point synthesis, which is key to our algorithmic approach in Section~\ref{sec:alglower}.
\end{remark}
}

Based on the discussion above, the problem of inferring upper (resp.~lower) bounds on $\vpf$ has now turned into the synthesis of pre (resp.~post) fixed-points in $\mathcal K^M$. In the next sections, we will focus on automated synthesis algorithms.

\section{Algorithmic Approaches to \uprobabbr}
\label{sec:algupper}

In this section, we consider the UQAVA problem and focus on inferring upper-bounds for the assertion violation probability. As mentioned in Section~\ref{sec:theory}, every pre fixed-point in $\mathcal K^M$ is an upper-bound on $\vpf$. We consider the case where the PTS $\pts$ and the invariant $\invmap$ are affine.

\paragraph{Affine \modelname} A \modelname{} $\pts$ is \emph{affine} if (i) every transition's guard condition is a conjunction of
affine inequalities over program variables, i.e.~inequalities of the form $\mathbf{a}^\mathrm{T}\cdot \vec{\vars} \le b$ where $\mathbf{a}^\mathrm{T}$ is a constant vector, $\vec{\vars}$ is the vector of program variables, and $b$ is a real number,
and (ii) every update function $\transupdate$ is affine, i.e.~ $\transupdate(\mathbf{v},\mathbf{u})=\mathbf{Q}\cdot \mathbf{v} + \mathbf{R}\cdot \mathbf{u} + \mathbf{e}$  where $\mathbf{Q}$ and $\mathbf{R}$ are constant matrices and $\mathbf{e}$ is a constant vector.

\paragraph{Affine Invariants} An invariant map $\invmap$ is \textit{affine} if for each $\loc\in \locs$, $\invmap(\locs)$ is a
	conjunction of affine inequalities over program variables.

We focus on synthesizing exponential upper-bounds (pre fixed-points). This choice best suits the common cases where the assertion violation probability decreases exponentially with respect to a combination of program variables. In general, due to transcendentality, exponential functions are much harder to synthesize than the widely-studied cases of linear functions~\cite{SriramCAV,DBLP:journals/toplas/ChatterjeeFNH18} or polynomials~\cite{ChatterjeeFG16}, which are respectively handled by Farkas' Lemma~\cite{FarkasLemma} and Positivstellens\"atze~\cite{ScheidererSurvey}. We present two algorithmic approaches for this problem:
\begin{itemize}
	\item In Section~\ref{sec:heuristicalg}, we show that Repulsing Ranking Supermartingales (RepRSMs), first defined in~\cite{ChatterjeeNZ2017} in the context of stochastic invariants, can be exploited to obtain exponential pre fixed-points. Our approach is based on Hoeffding's lemma and leads to an efficient sound algorithm that first synthesizes a linear/polynomial RepRSM, and then obtains an exponential pre fixed-point based on it. Our bounds are significantly better than the ones obtained in~\cite{ChatterjeeNZ2017} using Azuma's inequality (Remark~\ref{rem:stocinv}). However, this efficient algorithm is not complete.
	\item In Section~\ref{sec:soundcompalg}, we provide a sound and \emph{complete} algorithm for exponential bounds of the form $\exp(a \cdot \val + b)$ for affine PTSs. This algorithm depends on Minkowski decomposition. Hence, in theory, it is not as efficient as the one in Section~\ref{sec:heuristicalg}. However, it provides completeness guarantees and reduces the problem to convex optimization. In practice, it inherits the efficiency of convex optimization and \amir{easily handles various benchmarks} (Section~\ref{sec:evaluation}).
\end{itemize}

\subsection{A Sound Polynomial-time Algorithm}
\label{sec:heuristicalg}

To present our first synthesis algorithm, we define the notion of RepRSMs. The definition below is taken from~\cite{ChatterjeeNZ2017} and slightly modified to become applicable to PTSs.

\paragraph{RepRSMs}
A \emph{$(\ltrans,\dtrans,\epsilon)$-Repulsing Ranking Supermartingale} is
a Lebesgue-measurable, e.g.~linear or polynomial, function $\rterm:\pstates\to \mathbb R$ satisfying the following conditions:
\begin{enumerate}[leftmargin=2em]
\item[(C1)] $\rterm(\locin,\valin)\le 0$;
\item[(C2)] $\forall \val \in \mathbb R^\vars. \quad \val\models \invmap(\locf)\Rightarrow \rterm(\locf,\val)\ge 0$;
\item[(C3)] For every transition $\trans=(\frmloc,\transcond,\fork_1,\fork_2,\cdots,\fork_k)$,
where $\fork_j=\langle \toloc_j,p_j,\transupdate_j \rangle$, it holds that:
\begin{align*}
\forall \val \in &\mathbb R^\vars.\\ & \val \models\invmap(\frmloc)\land \transcond \Rightarrow \\
& \textstyle \sum_{j=1}^{k}
p_{j}\cdot
\mathbb{E}_{\valrd} [\rterm(\toloc_{j},\transupdate_{j}(\val,\valrd))]\le \rterm(\frmloc,\val)-\epsilon.
\end{align*}
\item[(C4)] For every $\trans$ as above and $1\le j\le k$, we have:
\begin{align*}
\forall \val \in \mathbb R^\vars, &\valrd \in \sampset. \\ \quad &\val \models \invmap(\frmloc) \wedge \transcond \Rightarrow\\
& \ltrans \le \rterm(\toloc_{j},\transupdate_{j}(\val,\valrd))-\rterm(\frmloc,\val)\le \ltrans+\dtrans.
\end{align*}
\end{enumerate}
Informally, (C1) says that the initial value of $\rterm$ is non-positive,  while (C2) means that when the program terminates with assertion failure, the value of $\rterm$ should be non-negative.
(C3) specifies that
the expected value of $\rterm$ decreases by at least $\epsilon$ after each transition in the \modelname{}. Finally,
(C4) states that the difference between current and next values of $\rterm$ always falls in the interval $[\ltrans,\ltrans+\dtrans]$.
In~\cite{ChatterjeeNZ2017}, it is shown that a RepRSM leads to an exponentially-decreasing upper-bound for assertion violation.
We now obtain a much tighter bound.

\begin{lemma}[Hoeffding's Lemma~\cite{ColinMcDiarmid1998concentration}]
\label{thm:hoef}
	For any random variable $X$ such that $a\leq X\le b$, and all $t \geq 0,$ we have
	$$\textstyle \mathbb E[\exp(t \cdot X)]\le \exp\left(t \cdot \mathbb E[X]+\frac{t^2 \cdot (b-a)^2}{8}\right).$$
\end{lemma}

We now present a theorem that establishes a connection between RepRSMs and pre fixed-point state functions, and serves as a basis for our first algorithm:

\begin{theorem}[Proof in Appendix \ref{appendix:algupper}]
\label{thm:reprsm}
Let $\rterm$ be a $(\ltrans,\dtrans,\epsilon)$-RepRSM, then
$\exp\left(\frac{8 \cdot \epsilon}{\dtrans^2} \cdot \rterm\right)$
is a pre fixed-point state function.
\end{theorem}

\begin{proof}[Proof Sketch]
Define $\usol:=\exp(\frac{8 \cdot \epsilon}{\dtrans^2} \cdot \eta)$. To prove $\pre(\usol)\\\sle \usol$, we expand the left-hand-side and directly upper-bound the exponential term using Lemma \ref{thm:hoef}.
\end{proof}

Combining the theorem above with our results in Section~\ref{sec:theory}, it is straightforward to see that $\exp(\frac{8 \cdot \epsilon}{\dtrans^2} \cdot \rterm(\locin, \valin))$ is an upper-bound on the probability of assertion violation.

\begin{remark}
\label{rem:stocinv}
Note that \cite{ChatterjeeNZ2017} also obtains an upper bound on assertion violation using RepRSMs. However, their method applies Azuma's inequality, which corresponds to the special case of $\ltrans<0 \land \dtrans = 2 \cdot \ltrans$ in our setting.
In their case,
given a $(-\dtrans/2,\dtrans,\epsilon)$-RepRSM $\rterm$,
the obtained bound is no less than
$\exp\left(\frac{4 \cdot \epsilon}{\dtrans^2} \cdot \eta\right),$
while our bound is $\exp(\frac{8 \cdot \epsilon}{\dtrans^2} \cdot \eta)$. Note that by condition (C1) in the definition of RepRSMs, $\eta(\locin, \valin)$ is non-positive. Thus, our bound is always substantially tighter.
\end{remark}

As shown by Theorem~\ref{thm:reprsm}, it is sufficient to synthesize a RepRSM in order to obtain an upper-bound for the assertion violation probability. In the rest of this section, we provide an algorithm that synthesizes linear RepRSMs over a given affine PTS $\pts$ with an affine invariant $\invmap$. This algorithm is a standard application of Farkas' lemma, as in~\cite{SriramCAV,DBLP:journals/toplas/ChatterjeeFNH18}. Hence, we only provide a high-level overview. See~\cite{SriramCAV,DBLP:journals/toplas/ChatterjeeFNH18} for a more detailed exposition. Finally, it is noteworthy that the algorithm can also be extended to polynomial RepRSMs (Remark~\ref{rem:polyheuristic}).

\paragraph{The \heurisalg Algorithm} Our algorithm derives an exponential upper-bound in four steps:

\paragraph{Step 1 (Setting up templates)}
The algorithm creates unknown coefficients $\vec{\mathbf{a}_{\loc}},b_{\loc}$ for every location $\loc \in \locs.$ Each
$\vec{\mathbf{a}_{\loc}}$ is a row vector of $\vert \vars \vert$ unknown coefficients and each $b_{\loc}$ is an unknown scalar. Moreover, the algorithm  symbolically computes $\rterm(\loc,\val):=\vec{\mathbf{a}_{\loc}} \cdot \val + b_{\loc}$ for every location $\loc$. The goal is to find values for the unknown coefficients $\vec{\mathbf{a}_{\loc}},b_{\loc},$ and RepRSM parameters $\beta, \Delta, \epsilon$ so that $\rterm$ becomes a RepRSM.

\paragraph{Step 2 (Collecting constraints)} The algorithm transforms (C2)--(C4) into conjunctions of constraints of the form
$$\textstyle \forall \mathbf{\val}\in P.\quad\left(\mathbf{c}^{\mathrm{T}}\cdot \mathbf{\val}\le d\right),$$ where
$P$ is a constant polyhedron and $\mathbf{c},d$ are, respectively, a vector and a scalar, with each of their component being an affine combination of the unknown coefficients created in the previous step. This step can be accomplished since both the $\pts$ and the invariant are affine.

\paragraph{Step 3 (Applying Farkas' Lemma)} Using Farkas' lemma, the algorithm transforms the constraints into an equivalent conjunctive collection of linear constraints over the unknowns.

\begin{lemma}[Farkas' Lemma~\cite{FarkasLemma}]
\label{thm:farkas}
Let $\mathbf{A}\in \mathbb R^{m\times n}$, $\mathbf{b}\in \mathbb R^{m}$, $c\in \mathbb R^{n}$ and $d\in \mathbb R$.
Assume that $P:=\{\mathbf{x}\in \mathbb{R}^n \mid \mathbf{A}\cdot\mathbf{x}\le \mathbf{b}\}\ne \emptyset$.
Then $P\subseteq \{\mathbf{x}\in \mathbb{R}^n \mid  \mathbf{c}^{\mathrm{T}}\cdot\val \le d \}$ iff there exists $\mathbf{y}\ge 0$ such that
$\mathbf{y}^{\mathrm{T}}\cdot \mathbf{A} =\mathbf{c}^{\mathrm{T}}$
and $\mathbf{y}^{\mathrm{T}}\cdot \mathbf{b}\le d$.
\end{lemma}

Every constraint of the previous step is of the form $\forall \val\in P\,\left(\mathbf{c}^{\mathrm{T}}\cdot \val\le d\right)$, which fits perfectly into the Farkas' Lemma.
Thus, by applying Farkas' Lemma, the algorithm obtains a linear programming instance over the unknown variables. Notably, no program variable appears in this linear program.

\paragraph{Step 4 ({Solving the unknown coefficients})} Our algorithm finds values for the unknown coefficients by solving the linear programming instance generated in the previous step together with the linear constraint from (C1).
Additionally, if the goal is to obtain the tightest possible upper-bound, rather than just any upper-bound, the algorithm instead solves the optimization problem with the objective of minimizing
$\frac{8 \cdot \epsilon}{\dtrans^2}\cdot\rterm(\locin,\valin)$. Finally, it obtains a RepRSM by plugging the obtained solution back into the template of Step 1, and an upper-bound on the assertion violation probability by simply computing $\frac{8 \cdot \epsilon}{\dtrans^2}\cdot\rterm(\locin,\valin).$

\begin{theorem}[Soundness]
    If \heurisalg successfully synthesizes $\rterm(\loc,\val)$, then the output $\exp(\frac{8 \cdot \epsilon}{\dtrans^2} \cdot \rterm(\locin,\valin))$ is an upper-bound on the probability of assertion violation in $\pts,$ given the initial state $(\locin,\valin)$.
\end{theorem}

\begin{proof}
It is easy to verify, by definition chasing, that our algorithm is sound and \emph{complete} for obtaining affine RepRSMs \cite{ChatterjeeNZ2017}
, since all steps reduce the problem to a new equivalent format.
The desired result is obtained by combining Theorem \ref{thm:reprsm} and the proof in Section~\ref{sec:theory} that every pre fixed-point is an upper-bound on $\vpf.$
\end{proof}

\begin{remark}[Extension to Polynomial Exponents]\label{rem:polyheuristic}
 The algorithm above handles the case where the exponent in our upper-bound is an affine combination of program variables. However,
it can be straightforwardly extended to polynomial exponents through Positivstelles\"{a}tze~\cite{ScheidererSurvey} and
semidefinite programming.
We obtain an exponential template with an affine exponent by directly synthesizing its affine exponent. This technique is also applicable to exponential templates with polynomial exponents, which are in turn obtained from polynomial RepRSMs.
We refer to~\cite{ChatterjeeFG16} for algorithmic details such as the use of Positivstelles\"{a}tze and semi-definite programming to synthesize polynomial (Rep)RSMs.
\end{remark}

\smallskip
\noindent{\bf \em Time Complexity.} The linear RepRSM synthesis takes polynomial time (via linear programming).
The same applies to polynomial RepRSMs~\cite{ChatterjeeFG16}. See Appendix~\ref{appendix:heuristicquadprog} for more details.

\subsection{A Sound and Complete Algorithm for Exponential Bounds with Affine Exponents}
\label{sec:soundcompalg}

We provide a sound and complete algorithm for the synthesis of upper-bounds which are of the form $\exp(\vec{a} \cdot \val + b).$ Our algorithm is based on Minkowski decomposition of polyhedra and a dedicated quantifier elimination method.

\paragraph{Polyhedra}
A subset $P\subseteq \mathbb{R}^n$ is a \emph{polyhedron} if $P=\{\mathbf{x}\in\mathbb{R}^n \mid \mathbf{A}\mathbf{x}\le \mathbf{b}\}$ for some matrix $\mathbf{A}\in\mathbb{R}^{m\times n}$ and vector $\mathbf{b}\in \mathbb{R}^n$.
A \textit{polytope} is a bounded polyhedron.
A \textit{cone} is a polyhedron $P$ such that $P=\{\mathbf{x}\in\mathbb{R}^n \mid \mathbf{A}\mathbf{x}\le \mathbf{0}\}$ for some matrix $\mathbf{A}\in\mathbb{R}^{m\times n}$.
A \emph{generator} set for a polytope $Q$ is a set of vectors $\{\val_1,\val_2,\cdots,\val_c\}$, such that every element $\val\in Q$ is representable as a convex combination of this set, i.e.~$\val=\sum_{i=1}^{c}{\lambda_i\cdot \val_i}$, for some coefficients $\lambda_i \geq 0$ with $\sum_{i=1}^{c}{\lambda_i}=1.$

\paragraph{Minkowski sum} Given two sets $A$ and $B$ of vectors, their Minkowski sum is defined as $A+B := \{x+y\ |\ x\in A,y\in B\}$.

\begin{theorem}[Decomposition Theorem~\cite{DBLP:books/daglib/0090562}]
\label{thm:pdt}
For every polyhedron $P,$ there exists a polytope $Q$ and a polyhedral cone $C,$ such that $P = Q+C.$
\end{theorem}

\paragraph{The \compalg\ Algorithm} Our algorithm takes as input an affine \modelnames $\pts$ and an affine invariant $\invmap$ for $\pts$.
If there exists an exponential pre fixed-point whose exponent is an affine expression over program variables, then it outputs such a function for $\pts$.
Otherwise, the algorithm asserts that there is no such state function.
It consists of five steps:

\paragraph{Step 1 (Setting up templates)}
The algorithm sets up a template $\usol$ as follows:
For each $\loc\not\in\{\loct,\locf\}$, it symbolically computes  $\usol(\loc, \val) :=\eterm{\loc}{\val}$
in which $\rterm(\loc,\val)$ is an affine function over program variables with unknown coefficients, i.e.~$\rterm(\loc,\val):=\vec{\mathbf{a}_{\loc}}\cdot \val + b_{\loc}$ where $\vec{\mathbf{a}_{\loc}}$ is a vector of unknown coefficients and $b_{\loc}$ is an unknown scalar.
Moreover, it sets $\usol(\loct,\val)\equiv 0$ and $\usol(\locf,\val)\equiv1.$ As in the previous section, our goal is to synthesize values for the unknown variables so that $\usol$ becomes a pre fixed-point (an upper-bound).

\begin{example}
\label{ex:soundstep1}
Consider the tortoise-hare example in Section \ref{sec:rabbit}, whose \modelnames is shown in Figure \ref{fig:pts21}.
For every location, we set up a vector of unknowns: $\vec{\mathbf{a}}_{\loc_{\rm init}}, \vec{\mathbf{a}}_{\loc_{\rm switch}}, \vec{\mathbf{a}}_{\loc_{\rm assert}}$.
We also set up unknown scalars $b_{\loc_{\rm init}},b_{\loc_{\rm assert}},b_{\loc_{\rm switch}}$.
We also compute $\usol$ and $\rterm$ symbolically, e.g.~
$$
\rterm(\loc_{\rm switch}, \val) = {a}_{\loc_{\rm switch},1} \cdot \val[x] + {a}_{\loc_{\rm switch},2} \cdot \val[y] + b_{\loc_{\rm switch}},
$$
$$
\usol(\loc_{\rm switch}, \val)  = \exp({a}_{\loc_{\rm switch},1} \cdot \val[x] + {a}_{\loc_{\rm switch},2} \cdot \val[y] +b_{\loc_{\rm switch}} ).
$$

\end{example}

\paragraph{Step 2 (Collecting constraints)}
The algorithm imposes pre fixed-point constraints to $\usol$.
Following the definition of $\pre$, for each transition
$\trans=(\frmloc,\transcond,\fork_1,\fork_2,\dots,\fork_k)$, where $\fork_j=\langle \toloc_j, p_j, \transupdate_j \rangle,$
the algorithm symbolically computes the following universally-quantified constraint and calls it $\cons^{\trans}$:
\begin{align}
\label{formula:cons}
	\textstyle	\forall \val \in \mathbb R^\vars.\quad &\val \models \invmap(\frmloc) \land \varphi\Rightarrow \\ \notag
		& \textstyle \usol(\frmloc,\val)\ge \sum_{j=1}^{k}{p_{j}\cdot {\textstyle \mathop{\mathbb E}_{\valrd}}[\usol(\toloc_{j},\transupdate_{j}(\val,\valrd))]}.
\end{align}
Intuitively, $\cons^{\trans}$ requires that for every valuation $\val$ that satisfies the invariant and the guard of transition $\trans,$ the pre fixed-point condition must be satisfied after going along $\trans$.

\begin{example}
	 Continuing with the previous example, consider the transition $\trans$ from $\loc_{\rm switch}$. The invariant for $\loc_{\rm switch}$ is $(x\le 99\land y\le 99)$ and the transition guard is $\transcond=\textbf{true}.$ Thus, every state $(\loc_{\rm switch},\val)$ that takes this transition must satisfy $(x\le 99\land y\le 99)$. The algorithm computes the pre fixed-point constraint $\cons^{\trans}$ as follows:
\begin{small}
$$
\forall x,y \in \mathbb R.\ (x\le 99\land y\le 99)\Rightarrow$$
$$
\usol(\loc_{\rm switch}, x, y)\ge 0.5\cdot \usol(\locin, x+1, y+2) + 0.5\cdot \usol(\locin, x+1, y).
$$
\end{small}
\end{example}

\paragraph{Step 3 (Canonicalization)} The algorithm transforms every constraint of Step 2 into the following \emph{canonical form}: $$\textstyle\forall \val \in \mathbb R^\vars.\quad (\val\in \guard)\Rightarrow \precond^\paracond(\val),$$ where $\guard$ is a polyhedron in $\mathbb R^\vars$ and $\precond^\paracond(\val)$ involves exponential terms on unknown coefficients and program variables, and is defined as:
\begin{align*}
&\textstyle\precond^\paracond(\val):=\\&\quad \quad\left(\sum_{j=1}^{k}{p_{j}\cdot \exp\left(\alpha_{j}\cdot \val + \beta_{j}\right)\cdot {\textstyle \mathop{\mathbb E}_{\valrd}} \left[\exp\left(\gamma_j\cdot \valrd\right)\right]}\le 1\right).
\end{align*}
Here, $\alpha_j,\beta_j,\gamma_j$ are affine expressions over unknown coefficients, and $p_j\in [0, +\infty)$.
We denote such a canonical constraint as $\cancond(\guard,\precond^\paracond)$. We now show how the algorithm canonicalizes every constraint of Step 2.
Consider the formula in \eqref{formula:cons}. The algorithm expands it based on the template:
	\begin{align*}
\textstyle		\exp(\rterm(\frmloc,\val))\ge\sum_{j=1}^{k}{p_{j}\cdot
{\textstyle \mathop{\mathbb E}_{\mathrm{\valrd}}}\left[\exp\left(\rterm(\toloc_{j},\transupdate_j(\val,\valrd))\right)\right]}
	\end{align*}
Now suppose that $\transupdate_j(\val,\valrd):=\mathbf{Q}_{j}  \val + \mathbf{R}_{j}  \valrd + \mathbf{e}_{j}.$
By further expanding both sides based on the template for $\rterm$ and dividing them by the left-hand-side, the algorithm obtains:
\begin{align*}
	\textstyle \sum_{j=1}^{k}{p_{j} \cdot \exp\left(\alpha_{j}\cdot \val + \beta_{j}\right) \cdot {\textstyle \mathop{\mathbb E}_{\mathrm{\valrd}}}\left[\exp\left(\gamma_j\cdot \valrd\right)\right]} \le 1
\end{align*}
where
$\alpha_{j}:=\mathbf{a}_{\toloc_{j}}\mathbf{Q}_{j} - \mathbf{a}_{\frmloc},\
\beta_{j}:=b_{\toloc_{j}} - \mathbf{a}_{\toloc_{j}}\cdot \mathbf{e}_{j} - b_{\frmloc},$ and
$\gamma_{j}:=\mathbf{a}_{\toloc_{j}}\mathbf{R}_{j}$
are affine expressions over the unknown coefficients.
\begin{example}
    Continuing with the previous example, by plugging in the template, the algorithm obtains:
\begin{small}
    \begin{align*}
    &\exp(a_{\loc_{\rm switch}}\cdot \val+b_{\loc_{\rm switch}})\ge\\& \quad \quad 0.5 \cdot  \exp(a_{\locin}\cdot (\val + [1\ 2]^T) + b_{\locin}) + \\ & \quad \quad 0.5 \cdot \exp(a_{\locin}\cdot (\val + [1\ 0]^T)+ b_{\locin})
    \end{align*}
\end{small}
which it then rewrites equivalently as:
\begin{small}
    \begin{align*}
    0.5 \cdot \exp(\alpha\cdot \val + \beta_1)+  \exp(\alpha\cdot \val + \beta_2) \le 1
    \end{align*}
     where
     $
     \alpha=\mathbf{a}_{\locin}-\mathbf{a}_{\loc_{\rm switch}},
     \beta_1=\mathbf{a}_{\locin}\cdot [1\ 2]^T + b_{\locin},
     \beta_2=\mathbf{a}_{\locin}\cdot [1\ 0]^T+ b_{\locin}
     $ are affine expressions. Let $\guard$ be the polyhedron defined by inequalities $x \leq 99 \land y \leq 99.$ Then, $\guard$ forms a canonical constraint together with the inequality above.
\end{small}
\end{example}

\paragraph{Step 4 (Quantifier Elimination)}
In this step, the algorithm eliminates the universal quantifier in the canonical constraints.
Our elimination technique relies on the decomposition theorem (Theorem \ref{thm:pdt})
to decompose a polyhedron
and Farkas' Lemma (Theorem \ref{thm:farkas}) to deal with linear constraints.

For each canonical constraint,
the algorithm first computes a decomposition of
$\guard$ as a Minkowski sum of a polytope $Q$ and a polyhedral cone $C$ and then transforms the constraint using the following proposition:

\begin{proposition}[Proof in Appendix \ref{appendix:algupper}]
\label{thm:cdt}
Given a canonical constraint $\cancond(\guard,\precond^\paracond)$, the polyhedron $\guard=\{\val\ |\ \mathbf{M}\val \le \mathbf{d}\}$ can be decomposed as $\guard=Q+C$, where $Q$ is a polytope and $C=\{\val\ |\ \mathbf{M}\val \le \mathbf{0}\}.$
Then,
 $\cancond(\guard,\precond^\paracond)$ is satisfied iff:
 \begin{itemize}[leftmargin=2em]
		\item[(D1)] $\forall j,\val.~~ \mathbf{C}\val \le \mathbf{0}\Rightarrow \alpha_j\cdot \val\le 0,$ and
		\item[(D2)] For every generator $\val^*$ of $Q$, $\val^*\models \precond^\paracond(\val^*)$.
	\end{itemize}
\end{proposition}
\begin{proof}[Proof Sketch]
For the if part, pick any $\val\in \guard$, by Theorem \ref{thm:pdt}, $\val=\val_1+\val_2$, where $\mathbf{M}\val_1\le \mathbf{0}$ and $\val_2\in Q.$ By plugging in $\val$ with $\val_1+\val_2$ into $\precond$ and the convexity of $\exp(\cdot)$, we prove that $\precond^\paracond(\val)$ holds. For the only if part, if $\cancond(\guard,\precond^\paracond)$ is satisfied, (D2) is true since $Q\subseteq \guard.$ We prove (D1) by contradiction. Suppose there exists $\val\in C$ and $j$ such that $\alpha_j\cdot \val>0$, choose any $\val_0\in Q$ and consider $\precond^\paracond(t\val+\val_0)$ for $t\ge 0$. By taking $t\to \infty$, $\precond^\paracond(t\val+\val_0)$ would be eventually violated, causing a contradiction.
\end{proof}

The algorithm computes (D1) and (D2). It translates (D1) to linear constraints using Lemma~\ref{thm:farkas}, and utilizes the double description method \cite{DBLP:conf/sas/BagnaraRZH02} to find all generators of $Q$ and write (D2) as a conjunction of finitely many convex inequalities.

\begin{example}
	Again, continuing with the previous example, the algorithm decomposes
$\guard$ into $\{(x,y)\ |\ x\le 0\land y\le 0\}$ and the polytope generated by a single point $\{(99,99)\}$ by the algorithm in \cite{DBLP:conf/sas/BagnaraRZH02}.
Then, it decomposes $\precond^{\paracond}$ into:
		$$\textstyle\forall x,y \in \mathbb R.~~ x\le 0\land y\le 0\Rightarrow \alpha\le 0, \text{and}$$
        $$\textstyle 0.5 \cdot\left(\exp(\alpha\cdot[99\ 99]^T + \beta_1)+  \exp(\alpha\cdot[99\ 99]^T + \beta_2)\right)\le 1$$
It then reduces the former to a conjunction of linear constraints on the unknown coefficients by Farkas' Lemma.
\end{example}

\paragraph{Step 5 (Optimization)}
After the steps above, all quantified canonical formulas are now reduced to a set of convex inequalities without quantifiers. To infer a tight upper-bound, our algorithm solves the following optimization problem $\Theta$ over the unknown variables defined in Step 1:
\begin{small}
\begin{equation}
\label{formula:opt}
\textstyle \textbf{Minimize~}  \exp(a_{\locin}\cdot \valin + b_{\locin})
\end{equation}
$$\textstyle \textbf{Subject to~} \text{the constraints obtained in Step 4 above.}$$
\end{small}
It calls an external solver to obtain the optimal solution.

\begin{theorem}[Proof in Appendix \ref{appendix:algupper}]\label{lem:opt}
$\Theta$ is a convex optimization problem.
\end{theorem}
\begin{proof}
Every constraint of Step 4 is either linear or of the form $L\le 1$ where $L$ is a non-negative combination of convex functions. This is due to the convexity of $\exp(\cdot).$
\end{proof}
So, we can use convex optimization tools to solve $\Theta.$
It is straightforward to verify the soundness of \compalg.
Our algorithm is also complete, formally:
\begin{theorem}
Given an affine \modelname, an affine invariant, and $\epsilon>0,$
\compalgs outputs an $\epsilon$-optimal solution for the unknown coefficients in an exponential template with affine exponent.
\end{theorem}
\begin{proof}
This follows from Proposition \ref{thm:cdt} that equivalently transforms the original synthesis into convex optimization.
\end{proof}

\hongfei{
\begin{remark}
The completeness is w.r.t exponential templates with affine exponent, i.e.~\compalgs can derive an approximately optimal template within any additive error bound.
Thus, the completeness is not related to decidability.
\end{remark}
}

\paragraph{Efficiency}
Theoretically, the costliest step of our algorithm is Step 3, because it requires the computation of decompositions of the polyhedra, which may cause exponential blow-up.
In practice, the constraint size after decomposition rarely explodes in real-world instances.
So, our algorithm inherits the well-known efficiency of convex programming.

\paragraph{Generality} Our algorithm is applicable to all cases in which we can expand ${\textstyle \mathbb E_{\valrd}}[\exp(\gamma_j\cdot \valrd)]=\prod_{\rdvar\in \rdvars}{{\textstyle {\mathbb E}_{r}}[\exp(\gamma_{j,r}\cdot r)]}$ into a simple closed form. Hence, it can handle all discrete distributions and many widely-used continuous distributions such as uniform distribution. For $r\sim \mathrm{uniform}[a,b]$, ${\textstyle {\mathbb E}_{r}}[\exp(\gamma_{j,r}\cdot r)]=\frac{1}{\gamma_{j,r}} \cdot (\exp({b\cdot \gamma_{j,r}})-\exp({a\cdot \gamma_{j,r}}))$ is the closed form.

\section{An Algorithmic Approach to \lprobabbr}
\label{sec:alglower}

In this section, we provide an efficient and automated algorithm for LQAVA over affine PTSs with affine invariants. Recall that in Section~\ref{sec:theory}, under the assumption of almost-sure termination, we succeeded in reducing the LQAVA problem to the synthesis of a post fixed-point (See Theorems~\ref{thm:tarski} and~\ref{thm:uniq}).

Unlike our algorithm for pre fixed-point synthesis (Section~\ref{sec:algupper}), exponential post fixed-point synthesis can no longer be transformed into convex optimization.
Instead, we propose a sound but incomplete algorithm that synthesizes exponential post fixed-point functions with linear exponents.
Our algorithm transforms the problem to linear programming by applying Jensen's inequality.

\begin{theorem}[Jensen's inequality~\cite{williams1991probability}]
\label{thm:jensen}
	For any convex function $f$ and random variable $X$, we have $\mathbb E[f(X)]\ge f(\mathbb E[X])$.
\end{theorem}

\paragraph{The \loweralg Algorithm} Our algorithm synthesizes an exponential lower-bound for assertion violation in five steps:

\paragraph{Step 1 (Setting up templates)}
Similar to our previous algorithms in Section~\ref{sec:algupper}, the algorithm introduces unknown coefficients $\vec{\mathbf{a}_{\loc}},b_{\loc}$ for every location $\loc\in \locs$, and symbolically computes the template $\lsol(\loc,\val):=\exp(\mathbf{a}_{\loc}\cdot \val+b_{\loc})$ for every $\loc \not\in \{\locf, \loct\}$. As usual, it lets $\lsol(\loct,\val)\equiv 0$ and $\lsol(\locf,\val)\equiv 1.$
\begin{example}
\label{ex:lowerstep1}
We now illustrate our algorithm on the program in Figure~\ref{fig:pts23}.\
As in Example \ref{ex:soundstep1}, the algorithm introduces unknown coefficients $\mathbf{a}_{\loc_{\rm init}}, \mathbf{a}_{\loc_{\rm switch}}, \mathbf{a}_{\loc_{\rm assert}}$ and $b_{\loc_{\rm init}},b_{\loc_{\rm assert}},b_{\loc_{\rm switch}}$.
\end{example}

\paragraph{Step 2 (Bounding)}
Note that Theorem~\ref{thm:uniq} requires that $\lsol$ be bounded.
To ensure this, the algorithm introduces a new unknown coefficient $M$ and generates the following constraint for every $\loc \in \locs$:
$$\textstyle \forall \val \in \mathbb R^\vars.~~ \val\models \invmap(\loc)\Rightarrow \mathbf{a}_{\loc}\cdot \val+b_{\loc}\le M.$$

\begin{example}
\label{ex:lowerstep2}
Continuing from Example \ref{ex:lowerstep1}, the algorithm introduces a new unknown coefficient $M$ and sets up a constraint for every location. For example, for $\locin$, whose invariant is $(x\le 100),$ it generates the following constraint:
$$\textstyle \forall x \in \mathbb R.~~x\le 100 \Rightarrow \mathbf{a}_{\locin}+b_{\locin}\cdot x\le M.$$
\end{example}

\paragraph{Step 3 (Collecting constraints and canonicalization)} Similar to Section~\ref{sec:soundcompalg},
the algorithm generates constraints that model the post fixed-point conditions.
For every transition,
the constraint is identical with Equation~\ref{formula:cons}, except that $\ge$ is replaced with $\le$. The algorithm then applies the same canonicalization as in Step 3 of Section~\ref{sec:soundcompalg}, i.e.~the constraint is transformed into the canonical form $Con(\guard,\postcond^\paracond)$, in which $\postcond^\paracond$ is $\precond^\paracond$ with $\le$ replaced by $\ge$.

\begin{example}
\label{ex:lowerstep3}
Continuing with Example \ref{ex:lowerstep2},
consider the unique transition out of $\loc_{\rm switch}.$
The algorithm generates
\begin{align*}
&\textstyle\forall x \in \mathbb R.~~ x\le 99\Rightarrow\\
&\textstyle\quad \lsol(\loc_{\rm switch},x)\le p\cdot \lsol(\loct,x) + 0.75(1-p) \cdot  \lsol(\locin,x+1)\\ &\textstyle\quad \quad \quad\quad\quad\quad \quad\quad+0.25(1-p) \cdot  \lsol(\locin,x-1).
\end{align*}
It writes this constraint in canonical form $\cancond(\guard,\postcond^\paracond),$ where $\guard=(x\le 99),$ and
$\postcond^\paracond$ is as follows:
$$\textstyle 0.25 \cdot (1-p) \cdot (3 \cdot \exp(\alpha\cdot x + \beta_1)+ \exp(\alpha\cdot x+\beta_2))\ge 1$$
in which $\alpha=\mathbf{a}_{\loc_{\rm init}}-\mathbf{a}_{\loc_{\rm switch}},\beta_1=\mathbf{a}_{\loc_{\rm init}}+b_{\loc_{\rm init}}$ and $\beta_2=b_{\loc_{\rm init}}-\mathbf{a}_{\loc_{\rm init}}$ are affine expressions over unknown coefficients.
\end{example}

\paragraph{Step 4 (Applying Jensen's inequality)}
Given a canonical constraint $\cancond(\guard,\postcond^\paracond)$, the algorithm strengthens $\postcond^\paracond$ to a conjunction of linear inequalities.
For $\postcond^\paracond$ as below:
$$
\textstyle \sum_{j=1}^{k}{p_{j}\cdot \exp\left(\alpha_{j}\cdot \val + \beta_{j}\right)\cdot {\textstyle \mathop{\mathbb E}_{\valrd}}\left[\exp\left(\gamma_j\cdot \valrd\right)\right]}\ge 1
$$
it divides both sides by $Q:={\textstyle\sum_{j=1}^k p_{j}}$, and applies Theorem~\ref{thm:jensen}, deriving the strengthened form $\relaxpostcond$:
$$
\textstyle Q^{-1} \cdot \sum_{j=1}^{k}{p_j} \cdot \left(\alpha_j\cdot \val+\beta_j+\gamma_j\cdot
{\textstyle \mathop{\mathbb E}_{\valrd}}\left[\valrd\right]\right)\ge -\ln Q
$$
\noindent Since $\alpha_j,\beta_j,\gamma_j$ are affine expressions, $\relaxpostcond$ is a linear inequality over our unknown coefficients. Note that this strengthening is sound but incomplete, i.e.~if $\relaxpostcond$ is satisfied, then so is $\postcond^\paracond$, because we can apply Jensen's inequality with exponential $f.$ However, the converse may not hold.

\begin{example} Continuing with Example \ref{ex:lowerstep3}, as per Jensen's inequality, the algorithm derives the strengthened form $\relaxpostcond$:
$$
\textstyle 0.75 \cdot (\alpha\cdot x + \beta_1)+ 0.25 \cdot (\alpha\cdot x+\beta_2)\ge -\ln (1-p).
$$
\end{example}

\paragraph{Step 5 (Farkas' Lemma and LP)}
The algorithm directly applies Lemma~\ref{thm:farkas} to convert every constraint generated in Steps 2 and 4 into an equivalent conjunction of linear constraints over the unknown coefficients.
Finally, it uses linear programming to solve these linear constraints. In order to obtain as tight a lower-bound as possible, the LP instance is solved with the objective of maximizing $\mathbf{a}_{\locin}\cdot \valin + b_{\locin}.$ Finally, the algorithm plugs the LP solutions back into the template and reports $\exp(\mathbf{a}_{\locin}\cdot \valin + b_{\locin})$ as the desired bound.

\begin{theorem}[Soundness]
\label{thm:soundlower}
	Given an affine almost-surely terminating \modelnames $\pts$  and an affine invariant $\invmap,$ the solution of the algorithm above is a bounded post fixed-point, and $\exp(\mathbf{a}_{\locin}\cdot \valin + b_{\locin})$ is a lower-bound on assertion violation probability.
\end{theorem}
\begin{proof}[Proof Sketch]
The constraints in Step 2 ensure the boundedness of $\lsol.$
By Theorem \ref{thm:jensen}, the strengthening in Step 4 is sound.
By (\ref{formula:upper}) in Theorem \ref{thm:tarski}, the desired result is obtained.
See Appendix \ref{appendix:alglower} for details.
\end{proof}

\paragraph{Complexity} We now analyze the time complexity of \loweralg.
In Step 5 we apply Farkas' Lemma, which takes polynomial time. It is straightforward to verify that the symbolic computations in all other steps take polynomial time, as well.
Finally, since LP can be solved in polynomial time, we conclude that our algorithm runs in polynomial time with respect to the size of the input \modelname{} and invariant.

\begin{remark}[The Polynomial Case]\label{rem:polylower}
Similar to Remark~\ref{rem:polyheuristic}, \loweralgs can also be extended to polynomial exponents via Positivstellens\"{a}tze and semidefinite programming.
\end{remark}

\section{Experimental Results}
\label{sec:evaluation}

\paragraph{Implementation} We implemented our algorithms in C++
and Matlab, and used PPL 1.2 \cite{DBLP:conf/sas/BagnaraRZH02} for Minkowski decompositions and CVX 2.2~\cite{cvx,gb08} for linear/convex optimization. All results were obtained on an Intel Core i7-8700K (3.7 GHz) machine with 32 GB of memory, running MS Windows 10.

\medskip
\paragraph{Benchmarks} We consider the following benchmarks from a variety of application domains~\cite{SriramCAV,DBLP:journals/toplas/ChatterjeeFNH18,ijcai18,pldi18,ChatterjeeNZ2017,DBLP:conf/oopsla/CarbinMR13,DBLP:journals/pacmpl/SmithHA19}:
\begin{itemize}[leftmargin=*]
	\item \textsc{Deviation}: In these benchmarks, taken from~\cite{SriramCAV}, the goal is to infer upper-bounds on the probability of large deviation of a program variable from its expected value upon termination. We compare the bounds obtained by our algorithms with those provided by~\cite{SriramCAV}.

	\item \textsc{Concentration}:  In this category, the goal is to derive upper-bounds on the probability that a probabilistic program does not terminate within a given number of steps. The programs are taken from~\cite{DBLP:journals/toplas/ChatterjeeFNH18,pldi18}, and we compare our results with those of~\cite{DBLP:journals/toplas/ChatterjeeFNH18}.
	\item \textsc{StoInv}: Stochastic invariants are closely related to and useful for deriving upper-bounds on the assertion violation probability. We take three benchmarks, namely \textsc{1DWalk}, \textsc{2DWalk}, and \textsc{3DWalk}, from \cite{ChatterjeeNZ2017} and also include our motivating example \textsc{Race} of Section~\ref{sec:rabbit}. We compare our derived upper-bounds with those of \cite{ChatterjeeNZ2017}.

	\item \textsc{Hardware}: These benchmarks require lower-bounds on the probability that a program run on unreliable hardware terminates successfully, i.e.~runs without errors until termination. The two benchmarks \textsc{Ref} and  \textsc{Newton} are taken from \cite{DBLP:conf/oopsla/CarbinMR13,DBLP:journals/pacmpl/SmithHA19}, whereas \textsc{M1Dwalk} is our motivating example in Section \ref{sec:hardware}. We made necessary abstractions to make the program fit into our framework, but we guarentee that the lower bound for abstracted program is also feasible for the original. When the data is available, we compare our derived lower-bounds with those from \cite{DBLP:conf/oopsla/CarbinMR13,DBLP:journals/pacmpl/SmithHA19}.
\end{itemize}
See Appendix \ref{appendix:evaluation} for details of benchmarks.

\medskip
\paragraph{Invariants and Termination} We manually derived affine invariants for the input PTSs. Alternatively, invariant generation, which is an orthogonal problem to ours, can be automated by approaches such as~\cite{Sting,DBLP:conf/cav/JeannetM09,DBLP:conf/pldi/Chatterjee0GG20,DBLP:conf/pldi/YaoRWJG20}. Similarly, we proved almost-sure termination by manually constructing ranking supermartingales~\cite{SriramCAV,DBLP:journals/toplas/ChatterjeeFNH18}. Proving almost-sure termination can also be automated by previous works such as~\cite{SriramCAV,DBLP:journals/toplas/ChatterjeeFNH18,ChatterjeeFG16,DBLP:conf/aplas/HuangFC18,mciver2017new}. 

\renewcommand{\arraystretch}{1.2}
\begin{table*}
	\vspace{-.5em}
\begin{footnotesize}
\begin{tabular}{|c|c|c||c|c||c|c||c|c|}
\hline
\multicolumn{2}{|c|}{\multirow{2}{*}{\textbf{Benchmark}}} & \multirow{2}{*}{\textbf{Parameters}} & \multicolumn{2}{c||}{\textbf{Algorithm of Section~\ref{sec:heuristicalg}}} & \multicolumn{2}{c||}{\textbf{Algorithm of Section~\ref{sec:soundcompalg}}} & \multirow{2}{*}{\textbf{Previous Results}} & \multirow{2}{*}{\textbf{Ratio}} \\ \cline{4-7}
\multicolumn{2}{|l|}{}                             &                             & \textbf{Upper-bound}          & \textbf{Time (s)}          & \textbf{Upper-bound}          & \textbf{Time (s)}          &                           &                        \\ \hline
\hline
\multirow{6}{*}{\begin{turn}{90}\textsc{Deviation}\end{turn}}
& \multirow{3}{*}{\textsc{RdAdder}}
& {$\Pr[X-\mathbb E[X]\ge 25]$} & $7.54\cdot 10^{-2}$  & $57.45$   & $7.43\cdot 10^{-2}$  & $0.95$    & $8.00 \cdot 10^{-2}$ & $1.07$\\
\cline{3-9}
& & {$\Pr[X-\mathbb E[X]\ge 50]$} & {$3.95\cdot 10^{-5}$}  & $58.05$ &  {$3.54\cdot 10^{-5}$} & $0.99$   & $4.54\cdot 10^{-5}$  & $1.28$\\
\cline{3-9}
& & {$\Pr[X-\mathbb E[X]\ge 75]$} & $1.44\cdot 10^{-10}$  & $57.45$   &  $9.17\cdot 10^{-11}$  & $0.91$   & $1.69\cdot 10^{-10}$ & $1.84$\\
\cline{2-9}
& \multirow{3}{*}{\textsc{Robot}}
& {$\Pr[X-\mathbb E[X]\ge 1.8]$} &  $1.66 \cdot 10^{-1}$ &  $127.00$  & $9.64\cdot 10^{-6}$   & $1.72$     & $2.04\cdot 10^{-5}$ & $2.11$ \\
\cline{3-9}
& & {$\Pr[X-\mathbb E[X]\ge 2.0]$} &  $6.81\cdot 10^{-3}$ & $124.02$   & $4.78\cdot 10^{-7}$   & $1.27$  & $1.62\cdot 10^{-6}$ & $3.39$ \\
\cline{3-9}
& & {$\Pr[X-\mathbb E[X]\ge 2.2]$} &  $5.66\cdot 10^{-5}$ &  $125.72$  &  $1.51\cdot 10^{-8}$  &  $1.24$ &  $9.85\cdot 10^{-8}$ & $6.52$\\
\hline
\multirow{9}{*}{\begin{turn}{90}\textsc{Concentration}\end{turn}}
& \multirow{3}{*}{\textsc{Coupon}}
& $\Pr[T>100]$ & $1.02 \cdot 10^{-1}$  &  $80.52$  &  $7.01\cdot 10^{-5}$  & $1.24$   &  $6.00 \cdot 10^{-3}$ & $85.59$ \\
\cline{3-9}
& & $\Pr[T>300]$ & $4.02\cdot 10^{-5}$  & $81.41$   &  $7.44\cdot 10^{-22}$  & $1.41$   & $9.01\cdot 10^{-10}$ & $1.21\cdot 10^{12}$\\
\cline{3-9}
& & $\Pr[T>500]$ & $1.40\cdot 10^{-8}$  & $80.80$   &  $4.01\cdot 10^{-40}$  & $1.23$   &  $1.05\cdot 10^{-16}$ & $2.61\cdot 10^{23}$\\
\cline{2-9}
& \multirow{3}{*}{\textsc{Prspeed}}
& $\Pr[T>150]$  & $5.42\cdot 10^{-7}$  &  $108.66$  & $7.43\cdot 10^{-23}$   &  $1.44$   & $5.00 \cdot 10^{-3}$ & $6.72\cdot 10^{19}$\\
\cline{3-9}
& & $\Pr[T>200]$ & $1.89\cdot 10^{-10}$ & $106.82$   &  $8.03\cdot 10^{-36}$  &  $1.19$  &  $2.59\cdot 10^{-5}$ & $3.23\cdot 10^{30}$\\
\cline{3-9}
& & $\Pr[T>250]$ & $5.65\cdot 10^{-14}$ & $108.09$   &  $2.71\cdot 10^{-49}$  & $1.09$     & $9.17\cdot 10^{-8}$ & $3.38\cdot 10^{41}$\\
\cline{2-9}
& \multirow{3}{*}{\textsc{Rdwalk}}
& $\Pr[T>400]$ &  $1.85\cdot 10^{-3}$ &  $44.44$  & $2.12\cdot 10^{-7}$   &  $0.55$   & $3.18\cdot 10^{-6}$ & $17.19$\\
\cline{3-9}
& & $\Pr[T>500]$ &  $1.43\cdot 10^{-5}$ &  $50.89$  &  $1.57\cdot 10^{-12}$  &  $0.58$ & $1.40\cdot 10^{-10}$& $89.17$\\
\cline{3-9}
& & $\Pr[T>600]$ &  $5.47\cdot 10^{-8}$ &  $49.16$  &  $4.81\cdot 10^{-18}$  & $0.66$ & $2.68\cdot 10^{-15}$ & $557.17$\\
\hline
\multirow{12}{*}{\begin{turn}{90}\textsc{StoInv}\end{turn}}
& \multirow{3}{*}{\textsc{1DWalk}}
& $x=10$ & $1.73\cdot 10^{-64}$  &  $48.44$  & $7.82\cdot 10^{-208}$ &  $1.19$  &  $5.1\cdot 10^{-5}$& $6.52\cdot 10^{202}$\\
\cline{3-9}
& & $x=50$ &  $6.77\cdot 10^{-62}$ &  $41.86$  &  $1.79\cdot 10^{-199}$  &  $1.08$   & $1.0\cdot 10^{-4}$ & $5.59\cdot 10^{194}$\\
\cline{3-9}
& & $x=100$ & $1.04\cdot 10^{-58}$ &  $41.18$  &  $5.03\cdot 10^{-189}$  &  $0.97$  &  $2.5\cdot 10^{-4}$ & $4.97\cdot 10^{184}$\\
\cline{2-9}
& \multirow{3}{*}{\textsc{2DWalk}}
& $(x,y)=(1000,10)$ & $4.14\cdot 10^{-73}$ & $53.69$   &  $1\cdot 10^{-655}$  &  $1.35$ & $2.4\cdot 10^{-11}$ & $2.4\cdot 10^{644}$\\
\cline{3-9}
& & $(x,y)=(500,40)$  & $6.43\cdot 10^{-37}$  &  $53.00$  &  $9.61\cdot 10^{-278}$  & $1.03$   &  $5.5\cdot 10^{-4}$ & $5.72\cdot 10^{273}$\\
\cline{3-9}
& & $(x,y)=(400,50)$ & $1.11\cdot 10^{-29}$ &  $52.58$  &  $1.02\cdot 10^{-218}$  & $1.37$   & $1.9\cdot 10^{-2}$ & $1.86\cdot 10^{216}$\\
\cline{2-9}
& \multirow{3}{*}{\textsc{3DWalk}}
& $\!(x,\!y,\!z)\!=\!(100,\!100,\!100)\!$ &  $4.83\cdot 10^{-281}$ & $85.07$   & $1\cdot 10^{-3230}$   & $1.20$   & $4.4\cdot 10^{-17}$ & $4.4\cdot 10^{3213}$ \\
\cline{3-9}
& & $\!(x,\!y,\!z)\!=\!(100,\!150,\!200)\!$ & $6.66\cdot 10^{-221}$ &  $84.86$  & $1\cdot 10^{-2538}$   & $1.25$  & $2.9\cdot 10^{-9}$ & $2.9\cdot 10^{2529}$\\
\cline{3-9}
& & $\!(x,\!y,\!z)\!=\!(300,\!100,\!150)\!$ & $7.86\cdot 10^{-181}$  & $83.28$   & $1\cdot 10^{-2076}$   & 1.37  & $1.3\cdot 10^{-7}$ & $1.3\cdot 10^{2069}$\\
\cline{2-9}
& \multirow{3}{*}{\textsc{Race}} & $(x,y)=(40,0)$ & $9.08\cdot 10^{-4}$  &  $55.24$  &  $1.52\cdot 10^{-7}$  & $0.89$  &  No result & --\\
\cline{3-9}
& & $(x,y)=(35,0)$ & $6.84\cdot 10^{-3}$  &  $54.23$  & $2.16\cdot 10^{-5}$   &  $0.78$  & No result & -- \\
\cline{3-9}
& & $(x,y)=(45,0)$ & $6.65\cdot 10^{-5}$  &  $56.39$  &  $8.65\cdot 10^{-11}$  & $0.67$   & No result & -- \\
\hline
\end{tabular}
\end{footnotesize}
\caption{Our Experimental Results for Upper-bound Benchmarks. The last column is $\frac{\text{previous bound}}{\text{our bound}}$.}
\label{table:exp}
\end{table*}
\begin{table*}
	\vspace{-1.7em}
\begin{footnotesize}
\begin{tabular}{|c|c|c||c|c||c|c|}
\hline
\multicolumn{2}{|c|}{\multirow{2}{*}{\textbf{Benchmark}}} & \multicolumn{1}{c||}{\multirow{2}{*}{\textbf{Parameters}}} & \multicolumn{2}{c||}{\textbf{Algorithm of Section~\ref{sec:alglower}}}                         & \multicolumn{1}{c|}{\multirow{2}{*}{\textbf{Previous Results}}} & \multicolumn{1}{c|}{\multirow{2}{*}{\textbf{Ratio}}} \\ \cline{4-5}
\multicolumn{2}{|c|}{}                           & \multicolumn{1}{c||}{}                            & \multicolumn{1}{c|}{\textbf{Lower-bound}} & \multicolumn{1}{c||}{\textbf{Time (s)}} & \multicolumn{1}{c|}{}                          & \multicolumn{1}{c|}{}                       \\ \hline \hline
\multirow{9}{*}{\begin{turn}{90}\textsc{Hardware}\end{turn}}
& \multirow{3}{*}{\textsc{M1DWalk}} & $p=10^{-7}$  & 0.999984 & 0.64 & Not applicable & -- \\
\cline{3-7}
& & $p=10^{-5}$  & $0.998401$ & $0.73$ &  Not applicable & --\\
\cline{3-7}
& & $p=10^{-4}$   & $0.984126$ & $0.54$ &  Not applicable & --\\
\cline{2-7}
 & \multirow{3}{*}{\textsc{Newton}} & $p=5\cdot 10^{-4}$  & $0.728492$ & $0.72$ &  No result & --\\
\cline{3-7}
& & $p=10^{-3}$ & $0.534989$ & $1.20$ & No result & --\\
\cline{3-7}
& & $p=1.5\cdot 10^{-3}$ & $0.392823$ & $0.67$ & No result & --\\
\cline{2-7}
& \multirow{3}{*}{\textsc{Ref}}
& $p=10^{-7}$ & 0.998463 & 1.03 & \makecell{$0.994885$ in \cite{DBLP:conf/oopsla/CarbinMR13}\\$0.992832$ in \cite{DBLP:journals/pacmpl/SmithHA19}} & \makecell{$3.33$\\$4.66$}\\
\cline{3-7}
& & $p=10^{-6}$ & 0.984738 & 1.03 &  No result & -- \\
\cline{3-7}
& & $p=10^{-5}$ & 0.857443 & 1.14 & No result & -- \\
\hline
\end{tabular}
\end{footnotesize}
\caption{Our Experimental Results for Lower-bound Benchmarks. The last column is $\frac{1-\text{previous bound}}{1-\text{our bound}}.$}
\vspace{-1.5em}
\label{table:exp2}
\end{table*}
\renewcommand{\arraystretch}{1}

\paragraph{Parameters}
Each benchmark set has distinct parameters:
For \textsc{Deviation} and \textsc{Concentration} the parameter is the deviation bound.
For \textsc{StoInv}, the parameters are the initial values of program variables.
For \textsc{Hardware}, the parameter is the probability of failure in each iteration.

\paragraph{Results} Our experimental results are summarized in Tables~\ref{table:exp} and~\ref{table:exp2}. ``No result'' means there is no previous experimental result reported and no available implementations to obtain such results. ``Not applicable'' means the benchmark is outside the theoretical framework of the previous work.
Note that in the \textsc{Hardware} examples, the data was only available for $p=10^{-7}$ in the literature and we could not find a public implementation of the approach.
See Appendix~\ref{appendix:evaluation} for more technical details.

\paragraph{Discussion} The experimental results show that our upper-bounds significantly beat the previous methods. Our algorithm from Section~\ref{sec:soundcompalg}, which is complete, consistently and significantly outperforms previous methods on all the benchmarks. The ratio of the bounds ranges from $1.07$ to $1.3 \cdot 10^{2069},$ i.e.~$2069$ orders of magnitude! Moreover, it achieves this in a maximum runtime of 1.72 seconds, which demonstrates its efficiency in practice. On the other hand,  our other algorithm (Section~\ref{sec:heuristicalg}), which is provably polynomial-time but not complete, synthesizes slightly looser bounds than~\cite{SriramCAV}  in a number of cases. We believe this is because \cite{SriramCAV} is specific to probabilistic programs with a fixed number of iterations, while our algorithm is applicable to general probabilistic programs.
In case of lower-bounds, we are providing the first automated algorithm. As such, there is very little data available from previous sources (i.e.~only for \textsc{Ref}). In this case, we also beat previous methods by a factor of $3.33$.

\section{Related Works} \label{sec:rel}

\paragraph{Probability Bounds for Assertion Violation}
This problem was first considered in~\cite{SriramCAV},
where it was shown that exponentially-decreasing upper bounds for the probability of large deviation from expected values can be derived through concentration inequalities and automatically generated using supermartingales.
Then,~\cite{DBLP:journals/toplas/ChatterjeeFNH18,ChatterjeeFG16} introduced a sound approach for deriving exponentially-decreasing upper bounds for the concencentration of termination time
through concentration inequalities, and developed automated algorithms through linear and polynomial ranking supermartingales.
For probabilistic programs that may not have exponentially decreasing concentration, sound approaches for deriving polynomial and square-root reciprocal upper bounds are introduced in ~\cite{DBLP:journals/corr/ChatterjeeF17, DBLP:conf/aplas/HuangFC18,DBLP:conf/tacas/KuraUH19,DBLP:journals/corr/abs-2001-10150}.
QAVA was formally proposed in~\cite{ChatterjeeNZ2017} as \emph{stochastic invariants},
where concentration inequalities were utilized to derive upper bounds for the probability of assertion violation and the synthesis of linear repulsive ranking supermartingales was adopted as the main algorithmic technique.
Later, probabilistic assertion violation analysis was considered as  \emph{accuracy analysis} in~\cite{DBLP:journals/pacmpl/SmithHA19} and an automata-based algorithm was proposed for loops with fixed number of iterations.
In our approach, we introduce novel fixed-point theorems for reasoning about both probability upper and lower bounds, and then develop new algorithmic techniques for synthesizing exponential templates that represent pre and post fixed-points.
Hence, compared with the above previous results, we have the following novelties:
\begin{itemize} [leftmargin=*]
\item our method is based on new insights in fixed-point theory rather than concentration inequalities or automata theory;
\item our approach derives both upper and lower bounds, while previous work only derive upper bounds;
\item we consider exponential templates that best match the situation where assertion violation probability decreases exponentially w.r.t certain amount.
\item we devise new algorithms for solving the exponential templates, including an algorithm that provides completeness in solving the template when the probabilistic program is affine and the exponent in the template is linear;
\item we prove in theory that the bounds generated by our approach is surely better than those from~\cite{ChatterjeeNZ2017};
\item the experimental results show that the bounds generated by our approaches are much better than previous results.
\end{itemize}

\paragraph{Expectation Bounds} There are also many results on expectation bounds for probabilistic programs, such as those based on fixed-point theorems~\cite{pldi18,DBLP:conf/pldi/Wang0GCQS19}, optional stopping theorems~\cite{ijcai18,DBLP:conf/pldi/Wang0GCQS19,ChatterjeeF19} and limit characterization~\cite{DBLP:journals/jacm/KaminskiKMO18,OLKMLICS2016}. Although assertion violation probabilities can be treated as expectation of indicator random variables that represent reachability to assertion violation, there are fundamental differences between our approach and these results.
\begin{itemize}[leftmargin=*]
\item Compared with the results using fixed-point theorems  (e.g. \cite{pldi18,DBLP:conf/pldi/Wang0GCQS19}),
the main strengths of our approach are: (i) we develop new fixed-point theorems that can derive both upper and lower bounds, while the classical least-fixed-point characterization only provides upper bounds;
and (ii) we consider exponential templates and devise algorithmic approaches that can solve them with completeness, while previous results only consider polynomial templates.
\item On the other hand, the results using optional stopping theorems (e.g.~\cite{ijcai18,DBLP:conf/pldi/Wang0GCQS19,ChatterjeeF19}) are difficult to apply to probability bounds of assertion violation. This is because in optional stopping theorems, one usually needs to interpret the random variable $X_T$ w.r.t a stochastic process $\Gamma=X_0, X_1, \dots$ and a stopping time $T$, but for assertion violation it is difficult to find a suitable interpretation for $X_T$ where the stochastic process
$\Gamma$ is typically defined by a template $\eta$ (i.e.~ $X_n:=\eta(\hat{v}_n)$ where $\hat{v}_n$ is the valuation at $n$-th step).

\item Finally, the results using limit characterization~\cite{DBLP:journals/jacm/KaminskiKMO18,OLKMLICS2016} require to build an infinite sequence of expressions that converges to certain limit. As such, they are difficult to automate. In contrast, our approach is entirely automated by constructing templates at each program counter and reducing the problem to optimization tasks.
\end{itemize}

\smallskip
\noindent{\textbf{Probability Bounds in Hybrid Systems.}} There are also several results that consider concentration bounds for hybrid systems~\cite{DBLP:journals/ijrr/SteinhardtT12,DBLP:conf/cav/FengC00Z20}.
~\cite{DBLP:journals/ijrr/SteinhardtT12} also considers the synthesis of exponential templates.
However, it only considers exponential templates in a very specific form, i.e~ the exponent is a positive semidefinite quadratic polynomial. In contrast, we use Hoeffding's Lemma and  Jensen's inequality to handle exponents in general form, and a novel convex optimization technique to completely solve the case that the both exponent and the underlying probabilistic program are affine.
\cite{DBLP:conf/cav/FengC00Z20} considers concentration bounds of stochastic differential equations and reduces the problem to semidefinite programming. Thus it is completely different from our approach.

\section{Conclusion and Future work}

In this work, we considered the problem of deriving quantitative bounds for assertion violation probabilities in probabilistic programs.
We established novel fixed-point theorems for upper and lower bounds on the assertion violation probability and presented three algorithms for deriving bounds in exponential form, one through RepRSMs and Hoeffding's Lemma, one through convex programming, and one through Jensen's inequality.
The experimental results show that our derived upper and lower bounds are much tighter than previous results.
An interesting direction for future work is to explore other, perhaps more expressive, forms of bounds.
Another future direction is to study compositional verification methods for bounding assertion violation probabilities.

\bibliography{PL}
\balance

\clearpage
\clubpenalty1000
\widowpenalty1000
\appendix
  \section{Appendix for Section \ref{sec:pts}}
\label{appendix:preliminary}

\subsection{Formal Definition of the \modelnames process}
\begin{definition}[\textit{\modelnames Process}]\label{def:ptsproc}
Let $\pts$ be a \modelname{}. Suppose that $\{\hat{\mathbf{u}}_n[r]\}_{n\ge 0, r\in\rdvars}$ is an independent collection of random variables such that each $\hat{\mathbf{u}}_n[r]$ is the random variable that observes the probability distribution $\rdvarjdis(r)$ and represents the sampled value for the sampling variable $r$ at the $n$th step.

The stochastic process $\trace$ induced by $\pts$ is a Markov process. It is an infinite sequence $\{\hat{\pstate}_n\}_{n\ge 0}$ of random variables such that (i) each $\hat{\pstate}_n$  equals $(\hat{\loc}_n, \hat{\val}_n)$ where $\hat{\loc}_n$ and $\hat{\val}_n$ are the random variables that represent the current location and resp. the current valuation for program variables at the $n$th step, and (ii) the random variables $\hat{\pstate}_n$ are inductively defined as follows:
\begin{itemize}
\item {\em Initial Step.} $\hat{\pstate}_0=(\hat{\loc}_{0}, \hat{\val}_{0}):=(\locin, \valin)$ (i.e.~a constant random variable).
\item {\em Inductive Step.} for each $n\ge 0$, we have $\hat{\pstate}_{n+1}=(\hat{\loc}_{n+1}, \hat{\val}_{n+1})$ where $(\hat{\loc}_{n+1}, \hat{\val}_{n+1})$ is defined as follows:
\begin{itemize}
\item if $\hat{\loc}_n\not\in \{\loct, \locf\}$, then we have exactly one transition $\langle\frmloc,\transcond,\fork_{1},\dots, \fork_{k}\rangle$ such that  $\hat{\loc}_n=\frmloc$ and $\hat{\val}_n\models\transcond$. In this case, a fork $\fork_{j}=\langle \toloc_{j}, \transprob_{j}, \transupdate_{j}\rangle$ is chosen with probability $\transprob_{j}$ and we have $(\hat{\loc}_{n+1}, \hat{\val}_{n+1}) = (\toloc_{j}, \transupdate_{j}(\hat{\val}_{n}, \hat{\mathbf{u}}_n))$;
\item if $\hat{\loc}_n\in \{\loct, \locf\}$ then the value of $(\hat{\loc}_{n+1}, \hat{\val}_{n+1})$ is taken to be the same as that of
    $(\hat{\loc}_n, \hat{\val}_n)$.
\end{itemize}
\end{itemize}
Note that the mutual-exclusiveness and completeness of transitions ensure that the stochastic process $\{\hat{\pstate}_n\}_{n\ge 0}$ is well-defined.
\end{definition}

  \section{Proofs of Section \ref{sec:theory}}
\label{appendix:theory}

We first establish some properties of $\mathcal K^M$.

\begin{proposition}
\label{prop:funcspace}
	For every $1\le M\le \infty$, $(\mathcal K^{M}, \sle)$ is a complete lattice. Furthermore, the smallest (resp.~greatest) element $\bot^{M}$ (resp.~$\top^M$) is the function whose value is \textit{0} (resp.~$M$) at all states $(\loc,\val)\in \mathcal S$ such that $\loc \notin \{\loct,\locf\}$.
\end{proposition}
\begin{proof}
    We show that every subset $\kappa \subseteq \mathcal K^M$ has an infimum and a supremum, thus $(\mathcal K^M, \sle)$ is a complete lattice.
	Fix any nonempty set $\kappa\subseteq \mathcal K^{M}$, define two functions $\kappa^{\sup}$ and $\kappa^{\inf}$:
	\begin{align*}
		\kappa^{\sup}(x)&:=\sup\left\{f(x)\ |\ f\in \kappa\right\}\\
		\kappa^{\inf}(x)&:=\inf\left\{f(x)\ |\ f\in \kappa\right\}
	\end{align*}
	We now verify that $\kappa^{\sup}$ and $\kappa^{\inf}$ are suprema and infima respectively, which directly follows from the definition of $\sup$ and $\inf$. In detail, by definition of $\sup$, $\forall f\in \kappa, \forall \pstate \in \pstates,  f(\pstate)\le \kappa^{\sup}(\pstate)$, thus $\kappa^{\sup}$ is an upper bound. Moreover, for any $g\in \mathcal K^{M}$. If $g$ is an upper bound of $S$, then for $\forall f\in \kappa	, \forall \pstate \in \pstates, f(\pstate)\le g(\pstate)$, then $\sup\{f(\pstate)\ |\ f\in S\}\le g(\pstate)$, then $\kappa^{\sup}(\pstate)\le g(\pstate)$. Hence $\kappa^{\sup}$ is the supremum. Similar for $\kappa^{\inf}$. Hence $(\mathcal K^M, \sle)$ is a complete lattice.

	By definition, $\top^{M}\in \mathcal K^{M}$, and for any function $f\in \mathcal K^{M}$, since $\forall \pstate\in \pstates, f(\pstate)\le M=\top^{M}(\pstate)$, hence $f\le \top^{M}$. Thus $\top^{M}$ is the greatest element. A similar argument handles the case of $\bot^{M}$.
\end{proof}

We now prove some propositions on the connection of $\pre$ and \modelname.

\begin{proposition}
\label{prop:pre}
For every $1\le M\le \infty$, $\pre^{M}: \mathcal K^{M}\to \mathcal K^{M}$ is a well-defined function. Furthermore, it is continuous for any $M$, and cocontinuous for finite $M$.
\end{proposition}
\begin{proof}
	Fix any $1\le M\le +\infty.$ We first prove that $\pre^{M}$ is well-defined. For every function $f\in \mathcal K^M$, we need to prove that for every $(\loc,\val)$, $\pre^M(f)(\loc,\val)\in [0,M]$. We do case analysis on $(\loc, \val)$:
	\begin{itemize}
		\item If $\loc = \locf$, then $\pre^M(f)(\loc,\val)=1\in [0,M]$.
		\item If $\loc = \loct$, then $\pre^M(f)(\loc,\val)=0\in [0,M]$.
        \item Otherwise, there is a unique transition $\trans=(\frmloc, \transcond, F_1,F_2,\\ \cdots, F_k)$ such that $\loc=\frmloc\land \val\models \transcond$, where the fork $F_j$ is $\langle \toloc_j, p_j, \transupdate_j\rangle$:
        \begin{align*}
			\pre^M(f)(\loc,\val) &= \sum_{j=1}^{k}{p_{j}\cdot {\textstyle \mathop{\mathbb E}_{\valrd}}[f(\toloc_{j},\transupdate_{j}(\val,\valrd))]} \\
			&\le \sum_{j=1}^{k}{p_{j}\cdot {\textstyle \mathop{\mathbb E}_{\valrd}}[M]} \\
			&= \sum_{j=1}^{k}{p_{j}\cdot M} \\
			&= M
		\end{align*}

	\end{itemize}

	Similarly, we can prove that $\pre^M(f)(\loc,\val)\ge 0.$ Thus, $\pre^M$ is well-defined. Now we prove that $\pre^M$ is monotone. Given any function $f,g$ such that $f\sle g$, by case analysis on $(\loc,\val)$:
	\begin{itemize}
		\item If $\loc = \locf$, then $$\pre^M(f)(\loc,\val)=1=\pre^M(g)(\loc,\val)$$
		\item If $\loc = \loct$, then $$\pre^M(f)(\loc,\val)=0=\pre^M(g)(\loc,\val)$$
		\item Otherwise, there is a unique transition $\trans=(\frmloc, \transcond, F_1,F_2,\\ \cdots, F_k)$ such that $\loc=\frmloc\land \val\models \transcond$, where the fork $F_j$ is $\langle \toloc_j, p_j, \transupdate_j\rangle$:
		\begin{align*}
			&\pre^M(f)(\loc,\val) \\
			&= \sum_{j=1}^{k}{p_{j}\cdot {\textstyle \mathop{\mathbb E}_{\valrd}}[f(\toloc_{j},\transupdate_{j}(\val,\valrd))]} \\
			&\le \sum_{j=1}^{k}{p_{i,j}\cdot {\textstyle \mathop{\mathbb E}_{\valrd}}[g(\toloc_{j},\transupdate_{j}(\val,\valrd))]} \\
			&=\pre^M(g)(\loc,\val)
		\end{align*}
	\end{itemize}

	Thus $\pre^M(f)\sle \pre^M(g)$, hence it is monotone.
	Next we prove upper continuity of $\pre^M.$ Choose any increasing chain $f_0\sle f_1\sle f_2\sle \cdots$ and do another case analysis on $(\loc,\val)$:
	\begin{itemize}
		\item If $\loc = \locf$, then $$\pre^M(\mathop{\bigsqcup}\limits_{n\ge 0}\left\{f_n\right\})(\loc,\val)=1=\mathop{\bigsqcup}\limits_{n\ge 0}\left\{\pre^M(f_n)\right\}(\loc,\val)$$
		\item If $\loc = \loct$, then $$\pre^M(\mathop{\bigsqcup}\limits_{n\ge 0}\{f_n\})(\loc,\val)=0=\mathop{\bigsqcup}\limits_{n\ge 0}\left\{\pre^M(f_n)\right\}(\loc,\val)$$
		\item Otherwise, there is a unique transition $\trans=(\frmloc, \transcond, F_1,F_2,\\ \cdots, F_k)$ such that $\loc=\frmloc\land \val\models \transcond$, where the fork $F_j$ is $\langle \toloc_j, p_j, \transupdate_j\rangle$:
		\begin{align*}
			&\pre^M(\mathop{\bigsqcup}\limits_{n\ge 0}\left\{f_n\right\})(\loc,\val) \\
			= &\sum_{j=1}^{k}{p_{j}\cdot {\textstyle \mathop{\mathbb E}_{\valrd}}\left[(\mathop{\bigsqcup}\limits_{n\ge 0}\left\{f_n\right\})(\toloc_{j},\transupdate_{j}(\val,\valrd))\right]} \\
			= &\sum_{j=1}^{k}{p_{j}\cdot {\textstyle \mathop{\mathbb E}_{\valrd}}\left[\sup_{n\ge 0}\left\{f_n(\toloc_{j},\transupdate_{j}(\val,\valrd))\right\}\right]} \\
		    = &\sum_{j=1}^{k}{p_{j}\cdot {\textstyle \mathop{\mathbb E}_{\valrd}}\left[\lim_{n\to \infty}\left\{f_n(\toloc_{j},\transupdate_{j}(\val,\valrd))\right\}\right]}\\
			\overset{\mathrm{MCT}}{=} &\sum_{j=1}^{k}{p_{j}\cdot \lim_{n\to \infty}{\textstyle \mathop{\mathbb E}_{\valrd}}\left[f_n(\toloc_{j},\transupdate_{j}(\val,\valrd))\right]}\\
			= &\lim_{n\to \infty}\sum_{j=1}^{k}{p_{j}\cdot{\textstyle \mathop{\mathbb E}_{\valrd}}\left[f_n(\toloc_{j},\transupdate_{j}(\val,\valrd))\right]}\\
			= &\lim_{n\to \infty}\pre^M(f_n)(\loc,\val)\\
			= &\sup_{n\ge 0}\left\{\pre^M(f_n)(\loc,\val)\right\}\\
			= &\mathop{\bigsqcup}\limits_{n\ge 0}\left\{\pre^M(f_n)\right\}(\loc,\val)
		\end{align*}
	\end{itemize}
	 The ``MCT'' above denotes the monotone convergence theorem. A similar argument establishes cocontinuity for finite $M$ and decreasing chains.
\end{proof}

\begin{proposition}
\label{prop:prepts}
Consider a \modelnames process $\hat{\sigma}_0,\hat{\sigma}_1,\hat{\sigma}_2,\cdots$. For every $n\ge 0, 1\le M\le +\infty$, and any function $f\in \mathcal K^M$, $\pre^{M}(f)(\hat{\sigma}_n)=\mathbb E[f(\hat{\sigma}_{n+1})\ |\ \hat{\sigma}_n]$
\end{proposition}
\begin{proof}
By definition, if $\hat{\loc}_n=\loct$, then $LHS=0$, and $\hat{\loc}_{n+1}=\hat{\loc}_n=\loct.$ Hence $f(\hat{\sigma}_{n+1})=0$ and $RHS=0=LHS$. The case for $\locf$ is similar. Otherwise, suppose at $n$th step, we choose the transition $\trans=(\frmloc, \transcond, F_1,F_2, \cdots, F_k)$ such that $\loc=\frmloc\land \val\models \transcond$, where the fork $F_j$ is $\langle \toloc_j, p_j, \transupdate_j\rangle$:
\begin{align*}
RHS&=\mathbb E\left[\sum_{j=1}^{k}p_{j}f(\toloc_{j},\transupdate_{i,j}(\hat{\val}_n,\valrd))\right]\\
&=\sum_{j=1}^{k}p_{j}\mathbb E_{\valrd}\left[f(\toloc_{j},\transupdate_{j}(\hat{\val}_n,\valrd))\right]\\
&=LHS
\end{align*}
\end{proof}

If we consider iteratively applying $\pre$ for $n$ times we derive the corollary below, which is useful in applying Theorem \ref{thm:kleene}:
\begin{corollary} \label{cor:comp}
	For any integer $n$, and any $1\leq M\leq \infty$,
$$\pre^{M,n}(f)(\loc,\val) = \mathbb E\left[f(\hat{\sigma}_n)\ |\ \hat{\sigma}_0=(\loc,\val)\right]$$
	where $\pre^{M,n}$ denotes the application of $\pre^M$ to $f$ for $n$ times. For $n=0$, we define $\pre^{M,0}(f):=f$.
\end{corollary}
\begin{proof}

We prove by induction.

\noindent{\textbf{Base case.}} For $n=0$, the result is obvious.

\noindent{\textbf{Induction case.}} Suppose the lemma holds fo $n=n_0$, we prove that it also holds for $n=n_0+1$.
\begin{align*}
LHS &= \mathbb E[\pre^M(f)(\hat{\sigma}_{n_0})\ |\ \hat{\sigma}_0=(\loc,\val)]\\
&=\mathbb E[\mathbb E[f(\hat{\sigma}_{n_0+1})\ |\ \hat{\sigma}_{n_0}]\ |\ \hat{\sigma}_0=(\loc,\val)]\\
&=\mathbb E[f(\hat{\sigma}_{n_0+1})\ |\ \hat{\sigma}_0=(\loc,\val)]
\end{align*}

The second equality is by Proposition \ref{prop:prepts}.
\end{proof}

\subsection{Proof of Theorem \ref{thm:epf}}
\begin{proof}

Fix any $1\le M\le \infty.$ by Proposition \ref{prop:pre}, $\pre$ is a continuous function. Now by Theorem \ref{thm:kleene}, we have:$$\lfp\ \pre^M=\mathop{\bigsqcup}\limits_{i\ge 0}\left\{\pre^{M,i}(\bot^M)\right\}.$$
Thus, for every $(\loc,\val)\in \pstates$: $$\lfp\ \pre^M(\loc,\val)=\mathop{\sup}\limits_{i\ge 0}\left\{\pre^{M,i}(\bot^M)(\loc,\val)\right\}.$$
We now apply Corollary \ref{cor:comp} to the $RHS$:
\begin{align*}
RHS&=\mathop{\sup}\limits_{i\ge 0}\left\{\mathbb E[\bot^M(\hat{\sigma}_i)\ |\ \hat{\sigma}_0=(\loc,\val)]\right\}\\
&=\mathop{\sup}\limits_{i\ge 0}\left\{\Pr[\hat{\loc}_i = \locf \ |\ \hat{\sigma}_0=(\loc,\val)]\right\}\\
&=\Pr[\exists i. \hat{\loc}_i = \locf \ |\ \hat{\sigma}_0=(\loc,\val)] = \epf(\loc,\val)\\
\end{align*}
\end{proof}

\subsection{Proof of Theorem \ref{thm:uniq}}
\begin{proof}

Fix any $1\le M< \infty.$ Since $M$ is finite, then by Proposition \ref{prop:pre}, $\pre^M$ is both continuous and cocontinuous. So, by Theorem \ref{thm:kleene},
\begin{align*}
\lfp\ \pre^M(\loc,\val)&=\mathop{\sup}\limits_{i\ge 0}\left\{\pre^{M,i}(\bot^M)(\loc,\val)\right\}\\
\gfp\ \pre^M(\loc,\val)&=\mathop{\inf}\limits_{i\ge 0}\left\{\pre^{M,i}(\top^M)(\loc,\val)\right\}
\end{align*}

Now by Proposition \ref{prop:funcspace} and Corollary \ref{cor:comp}, we can plug in the concrete form of $\bot$ and $\top$, and expand the right side:
\begin{align*}
\lfp\ \pre^M(\loc,\val)&=\sup_{n\ge 0}\{\Pr[\hat{\loc}_n=\locf]\}\\
\gfp\ \pre^M(\loc,\val)&=\inf_{n\ge 0}\{M\Pr[\hat{\loc}_n\neq \loct\land \hat{\loc}_n\neq \locf]+ \Pr[\hat{\loc}_n=\locf]\}
\end{align*}
Thus, for every $n$:
\begin{align*}
&\gfp\ \pre^M(\loc,\val) - \lfp\ \pre^M(\loc,\val)\\&\le \inf_{n\ge 0}\{M\cdot \Pr[\hat{\loc}_n\neq \loct\land \hat{\loc}_n\neq \locf]\}\\
&=0
\end{align*}
Thus, $\gfp\ \pre^M = \lfp\ \pre^M$, combined with Theorem \ref{thm:epf}, we derive this theorem.
\end{proof}

  \section{Proofs of Section \ref{sec:algupper}}
\label{appendix:algupper}
\subsection{Proof of Theorem \ref{thm:reprsm}}
Define $\usol(\loc,\val):=\exp(\frac{8\epsilon}{\dtrans^2}\rterm'(\loc,\val))$ and $\alpha:=\frac{8\epsilon}{\dtrans^2}$.
First, we prove that $\usol\in \mathcal K^{+\infty}$. By construction, for every $\val$, $\usol(\loct,\val)=0\land \usol(\locf,\val)=1.$ Thus, $\usol\in \mathcal K^{+\infty}$.
Now, by the definition of pre fixed-point, we need to prove $\pre^{+\infty}(\usol)\sle \usol$. We prove this by case analysis:

For $\loc=\loct$, $\pre^{+\infty}(\usol)(\loct,\val)=0\le \usol(\loct,\val)$.

For $\loc=\locf$, $\pre^{+\infty}(\usol)(\locf,\val)=1\le \usol(\locf,\val)$.

Otherwise, for every state $(\loc,\val)$ with $\loc \notin \{\loct,\locf\}$, suppose it would transit along $\trans_i$ with $k$ forks:
\begin{small}
\begin{align*}
&\frac{\pre^{+\infty}(\usol)(\loc,\val)}{\usol(\loc,\val)}\\
&=\frac{1}{\usol(\loc,\val)}\sum_{j=1}^{k}{p_{i,j}\mathop{\mathbb E}\limits_{\mathbf{\hat{u}}\sim \mathcal D}\left[\usol(\toloc_{i,j},F_{i,j}(\val,\mathbf{\hat{u}}))\right]}\\
&= \sum_{j=1}^{k}{p_{i,j}\mathop{\mathbb E}\limits_{\mathbf{\hat{u}}\sim \mathcal D}\left[\exp\left(\alpha\left(\eta'(\toloc_{i,j},F_{i,j}(\val,\mathbf{\hat{u}}))-\eta'(\frmloc_i,\val)\right)\right)\right]}\\
&\le \sum_{j=1}^{k}{p_{i,j}\mathop{\mathbb E}\limits_{\mathbf{\hat{u}}\sim \mathcal D}\left[\exp\left(\alpha\left(\eta(\toloc_{i,j},F_{i,j}(\val,\mathbf{\hat{u}}))-\eta(\frmloc_i,\val)\right)\right)\right]}
\end{align*}
\end{small}
We further upper bound the expectation by Hoeffding's Lemma:
\begin{small}
\begin{align*}
&\mathop{\mathbb E}\limits_{\mathbf{\hat{u}}\sim \mathcal D}\left[\exp\left(\alpha\left(\eta(\toloc_{i,j},F_{i,j}(\val,\mathbf{\hat{u}}))-\eta(\frmloc_i,\val)\right)\right)\right]\\
&\le \exp\left(\mathop{\mathbb E}\limits_{\mathbf{\hat{u}}\sim \mathcal D}\left[\alpha\left(\eta(\toloc_{i,j},F_{i,j}(\val,\mathbf{\hat{u}}))-\eta(\frmloc_i,\val)\right)\right]+\frac{(\alpha\cdot \dtrans)^2}{8}\right)\\
&\le \exp\left(\alpha\cdot (-\epsilon)+\frac{(\alpha\cdot \dtrans)^2}{8}\right)\\
&\le \exp\left(-\frac{8\epsilon^2}{\dtrans^2}+\frac{8\epsilon^2}{\dtrans^2}\right)=1
\end{align*}
\end{small}
Thus,
\begin{small}
\begin{align*}
\frac{\pre^{+\infty}(\usol)(\loc,\val)}{\usol(\loc,\val)}\le \sum_{j=1}^{k}{p_{i,j}}=1
\end{align*}
\begin{align*}
\pre^{+\infty}(\usol)(\loc,\val)\le \usol(\loc,\val)
\end{align*}
\end{small}

So, $\pre^{+\infty}(\usol)\sle \usol$, and we conclude that $\usol$ is a pre fixed-point.

\subsection{Details of Quadratic Programming}
\label{appendix:heuristicquadprog}
In this section, we introduce our algorithm for solving the quadratic programming problem in Section \ref{sec:heuristicalg}.

Before demonstrating our algorithm, we first establish some properties of this optimization problem.
Recall the original optimization problem:
$$\min\quad \quad \frac{8\epsilon}{\dtrans^2}\eta(\locin,\valin)$$
such that (C1)--(C4) are satisfied. Furthermore, we need to restrict $\epsilon\ge 0, \dtrans>0, \eta(\locin,\valin)\le 0$.

First, we simplify the problem. Since we can scale $\eta$, it suffices to find a solution with $\dtrans=1$.
Next, we introduce a fresh unknown coefficient $\omega$ and add the constraint $0\ge \omega\ge \mathbf{a}_{\locin}\cdot \valin + b_{\locin}$.
Finally, we modify the objective function to $8\cdot \epsilon\cdot \omega$.
Since we need to minimize the objective function, the original problem is equivalent to the simplified one.

After simplification, this optimization problem is an instance of quadratic programming, since the objective function is the multiplication of two unknown coefficients and all constraints are linear.
However, it is not an instance of convex programming because the objective function is not convex. This being said, we can still prove the uniqueness of local optimum.

For simplicity, in the rest of this section we use $\lambda$ to represent all unknown coefficients other than $\epsilon$ and $\omega$, including $\beta,\delta,\mathbf{a}'s,b's$ in the original problem.

\begin{proposition}
If the optimization problem has a solution with $\epsilon>0$ and $\omega<0$, then the local optimum is unique.
\end{proposition}
\begin{proof}

Since there is a solution with $\omega<0$ and $\epsilon>0$, there is no local optimum with $\omega=0$ or $\epsilon=0$.
After wiping out the case of $\omega=0$ or $\epsilon=0$, we prove this proposition by contradiction.
Suppose there are two different local optima $(\epsilon_1,\omega_1,\lambda_1)$ and $(\epsilon_2,\omega_2,\lambda_2)$, with $\epsilon_1,\epsilon_2>0\land \omega_1,\omega_2<0$. We prove $\epsilon_1=\epsilon_2\land \omega_1 =  \omega_2$.

Without loss of generality, we assume $\epsilon_1\cdot \omega_1 \ge \epsilon_2\cdot \omega_2$. Then there exists $0\ge \omega'\ge\omega_2$ such that $\epsilon_2\cdot \omega'=\epsilon_1\cdot \omega_1$.
By construction $(\epsilon_2,\omega',\lambda_2)$ is still a feasible solution.
Since every constraint is linear, for all $t\in [0,1]$, $(t\cdot \epsilon_1+(1-t)\cdot \epsilon_2,t\cdot \omega_1+(1-t)\cdot \omega',t\cdot \lambda_1+(1-t)\cdot \lambda_2)$ is a feasible solution, whose objective value is $\epsilon_1\cdot \omega_1+\left(2t^2-2t\right)\cdot \epsilon_1\cdot \omega_1 + t\cdot (1-t)\cdot (\epsilon_1\cdot \omega'+\epsilon_2\cdot \omega_1)$, which is strictly less than $\epsilon_1\cdot \omega_1$ for all $t\in (0,1)$. This is derived by the following calculation:
\begin{small}
\begin{align*}
&\left(2t^2-2t\right)\cdot \epsilon_1\cdot \omega_1 + t\cdot (1-t)\cdot (\epsilon_1\cdot \omega'+\epsilon_2\cdot \omega_1)\\
&=(t^2-t)\cdot (\epsilon_1-\epsilon_2)\cdot (\omega_1-\omega') < 0
\end{align*}
\end{small}
 Hence, $(\epsilon_1,\omega_1,\lambda_1)$ is not a local optimum, causing a contradiction.
\end{proof}

Suppose the unique local optimum is $(\epsilon^*,\omega^*,\lambda^*)$. By uniqueness, it is also a global optimum.
Now we can take a different perspective of this optimization problem. We regard this optimization problem as finding the minimum value of a function of $\epsilon$. Fix $\epsilon$ and define $(\epsilon,\omega^{\rm opt}(\epsilon), \lambda^{\rm opt}(\epsilon))$ as the choice that minimizes the objective function under $\epsilon$. If the optimization problem is infeasible under $\epsilon$, we define $\omega^{\rm opt}(\epsilon):=\infty$. We further define $f(\epsilon):=8\cdot \epsilon\cdot \omega^{\rm opt}(\epsilon)$ as the optimal value of objective function under fixed $\epsilon$.

Since all constraints become stricter as $\epsilon$ increases, $\omega^{\rm opt}(\epsilon)$ is a monotonically increasing function of $\epsilon.$ We define $\epsilon_{\rm max}:=\inf\{\epsilon\ |\ \omega^{\rm opt}(\epsilon)<\infty\}$. We further show that $f(\epsilon)$ first strictly decreases and then strictly increases. Formally:
\begin{proposition}
\label{prop:f}
In $[0,\epsilon^*]$, $f(\epsilon)$ is a strictly decreasing function, and it is a strictly increasing function in $[\epsilon^*,\epsilon_{\rm max}]$.
\end{proposition}
\begin{proof}
We only prove the case for $\epsilon\in [0,\epsilon^*]$, the other case is similar.
Arbitrarily pick $0\le \epsilon_1<\epsilon_2\le \epsilon^*.$ We need to prove that $f(\epsilon_1)>f(\epsilon_2)$. If $\epsilon_2=\epsilon^*$, by uniqueness of optimum, $f(\epsilon_1)<f(\epsilon^*)=f(\epsilon_2)$. If $\epsilon_2<\epsilon^*$, since $f(\epsilon_1)<f(\epsilon^*)$, we can take $0\ge \omega'>\omega^*$ such that $f(\epsilon_1)=\omega'\cdot \epsilon^*$, and by construction $(\epsilon^*, \omega', \lambda^*)$ is also a feasible solutions. By linearity of constraints, for all $t\in [0,1]$, $(t\cdot \epsilon_1+(1-t)\cdot \epsilon^*,t\cdot \omega^{\rm opt}(\epsilon_1)+(1-t)\cdot \omega',t\cdot \lambda^{\rm opt}(\epsilon_1)+(1-t)\cdot \lambda^*)$ is a feasible solution, whose objective value is $f(\epsilon_1)+\left(2t^2-2t\right)\cdot f(\epsilon_1) + t\cdot (1-t)\cdot (\epsilon_1\cdot \omega'+\epsilon^*\cdot \omega^{\rm opt}(\epsilon_1))$, which is strictly smaller than $f(\epsilon_1).$ Plugging in $t_0=\frac{\epsilon^*-\epsilon_2}{\epsilon^*-\epsilon_1}\in (0,1)$, we get that $(\epsilon_2,t_0\cdot \omega^{\rm opt}{\epsilon_1}+(1-t_0)\cdot \omega',t_0\cdot \lambda^{\rm opt}(\epsilon_1)+(1-t_0)\cdot \lambda^*)$ is a feasible solution. Thus, $f(\epsilon_1)> \epsilon_2\cdot (t_0\cdot \omega^{\rm opt}{\epsilon_1}+(1-t_0)\cdot \omega') \ge \epsilon_2\cdot \omega^{\rm opt}(\epsilon_2)=f(\epsilon_2)$.
\end{proof}

Now we describe our algorithm \textsf{Ser}. \textsf{Ser} takes an optimization problem in the form of Section \ref{sec:heuristicalg} and outputs a solution to the optimization problem.

\smallskip
\noindent{\textbf{Step 1. Feasibility checking. }} In this step, we first check whether $\epsilon_{\rm max}=\infty$. This can be achieved through linear programming, where the objective is to maximize $\epsilon$ and the constraints are the same. If $\epsilon_{\rm max}=\infty$, we can simply output $0$ and terminate the algorithm. Then, we check whether there exists a solution with $\epsilon>0$ and $\omega<0$, which is equivalent to $\epsilon_{\rm max}>0$ and $\omega^{\rm opt}(0)<0$. This is also achievable by linear programming. If there is no such solution, we simply output $1$ and terminate the algorithm.

\smallskip
\noindent{\textbf{Step 2. Search for $\epsilon^*$. }} If the algorithm does not terminate after Step 1, we know that the global optimum is unique. Then, we iteratively search $\epsilon^*$, the iteration procedure starts with $l=0, r=\epsilon_{\rm max}$, every time in iteration, our algorithm calculates $m_1=\frac13(2l+r),m_2=\frac13(l+2r)$, if $f(m_1)<f(m_2)$, then it sets $r\leftarrow m_2$ else it sets $l\leftarrow m_1$ and then repeats the iteration again. Note that we can solve $f(\epsilon),\omega^{\rm opt}(\epsilon)$ and $\lambda^{\rm opt}(\epsilon)$ by linear programming, since when $\epsilon$ is fixed, both the objective function and the constraints are linear. The iteration stops when $r-l<\mu$, where $\mu$ is a given error bound. Finally, the algorithm outputs $l$, $\omega^{\rm opt}(l)$ and $\lambda^{\rm opt}(l)$ as the final solution.

Our algorithm can efficiently approach the global optimum with arbitrarily small error. Formally, we have the following theorem:
\begin{theorem}
Given the optimization problem in the form of Section \ref{sec:heuristicalg}, and any error bound $\mu>0,$ \textsf{Ser}  outputs a solution $(\epsilon_0, \omega_0, \lambda_0)$ such that $|\epsilon_0-\epsilon^*|\le \mu$ in $O(\log \frac{\epsilon_{\rm max}}{\mu})$ within iterations of applying linear programming.
\end{theorem}
\begin{proof}
We first prove that in any iteration, $\epsilon^*\in [l,r]$. Hence, when the iteration stops, we derive that $|l-\epsilon^*|\le |l-r|\le\mu$. We prove this by induction.

\smallskip
\noindent{\textbf{Base case.}} At the start of iteration, $l=0\land r=\epsilon_{\rm max}.$ It is obvious that $\epsilon^*\in [l,r]$.

\smallskip
\noindent{\textbf{Inductive step.}} In each iteration, suppose $\epsilon^*\in [l,r].$ If $f(m_1)<f(m_2),$ then we claim that $\epsilon^*\notin [m_2,r]$. Otherwise $m_1<m_2<\epsilon^*<r$. By Proposition \ref{prop:f}, we must have $f(m_1)>f(m_2)$, causing a contradiction. Hence, $\epsilon^*\in [l,m_2]$ and the induction succeeds. Similarly, if $f(m_1)>f(m_2)$, we can also prove that $\epsilon^*\in [m_1,r]$.

Suppose we iterate $q$ times, then $r-l$ after $q$ iterations would be $\left(\frac23\right)^q\epsilon_{\rm max}$. Since the iteration stops when $r-l<\mu$, we conclude that $q=O(\log \frac{\epsilon_{\rm max}}{\mu})$. Since in each iteration we solve $O(1)$ linear programming instances, we derive the theorem.
\end{proof}

  \subsection{Proof of Proposition \ref{thm:cdt}}
Consider the canonical constraint $\cancond(\guard,\precond^\paracond)$, where $\guard=(\mathbf{M}\val\le \mathbf{d})$. Then, $\guard=Q+C$, where $Q$ is a polytope and $C=\{\val\ |\ \mathbf{M}\val \le \mathbf{0}\}$. Suppose $\val_1^*,\val_2^*\cdots, \val_c^*$ are generators of $Q$.

For the if part, pick any $\val\models \guard.$ By Theorem \ref{thm:pdt}, $\val=\val_1+\val_2$, where $\mathbf{M}\val_1\le \mathbf{0}$ and $\val_2\in Q$. We plug in $\val$ with $\val_1+\val_2$ into the $LHS$ of $\precond^\paracond(\val)$:
\begin{small}
\begin{align*}
LHS&=\sum_{j=1}^{k}{p_{j} \exp\left(\alpha_{j}\cdot (\val_1+\val_2) + \beta_{j}\right) \mathop{\mathbb E}\limits_{\mathrm{\mathbf{\hat{u}}}}\left[\exp\left(\gamma_j\cdot \mathbf{\hat{u}}\right)\right]}
\\&\le \sum_{j=1}^{k}{p_{j} \exp\left(\alpha_{j}\cdot \val_2 + \beta_{j}\right) \mathop{\mathbb E}\limits_{\mathrm{\mathbf{\hat{u}}}}\left[\exp\left(\gamma_j\cdot \mathbf{\hat{u}}\right)\right]}
\end{align*}
\end{small}

The first $\le$ is derived by (D1). Since $\val_2\in Q$, $\val_2$ can be represented as a convex combination over generators: $\val_2=\sum_{1\le i\le c}{\lambda_i\cdot \val_i^*}$, where $\lambda_i\ge 0$ and $\sum_{1\le i\le c}\lambda_i=1$.
By Jensen's inequality (Theorem \ref{thm:jensen}), we can further upper bound:
\begin{small}
\begin{align*}
LHS&\le \sum_{i=1}^{c}\lambda_i\sum_{j=1}^{k}{p_{j} \exp\left(\alpha_{j}\cdot \val_i^* + \beta_{j}\right) \mathop{\mathbb E}\limits_{\mathrm{\mathbf{\hat{u}}}}\left[\exp\left(\gamma_j\cdot \mathbf{\hat{u}}\right)\right]}\\
&\le \max_{1\le i\le c}\left\{\sum_{j=1}^{k}{p_{j} \exp\left(\alpha_{j}\cdot \val_i^* + \beta_{j}\right) \mathop{\mathbb E}\limits_{\mathrm{\mathbf{\hat{u}}}}\left[\exp\left(\gamma_j\cdot \mathbf{\hat{u}}\right)\right]}\right\}
\end{align*}
\end{small}
By (D2), $LHS\le 1$, thus $\precond^\paracond(\val)$ is true.

For the only if part, if $\cancond(\guard,\precond^\paracond)$ is satisfied,
(D2) is true since $Q\subseteq \guard$. We prove (D1) by contradiction. Suppose there exists $\val\in C$ and $j^*$ such that $\alpha_{j^*}\cdot \val>0.$ Choose some element $\val_0\in Q$ and consider the $LHS$ of $\precond^\paracond(t\val+\val_0)$ for $t\ge 0$ (Note that since $\val\in C$ and $\val_0\in Q$, thus $t\val \in C$ and $t\val + \val_0\models \guard$):
\begin{small}
\begin{align*}
LHS&=\sum_{j=1}^{k}{p_{j} \exp\left(\alpha_{j}\cdot (t\val+\val_0) + \beta_{j}\right) \mathop{\mathbb E}\limits_{\mathrm{\mathbf{\hat{u}}}}\left[\exp\left(\gamma_j\cdot \mathbf{\hat{u}}\right)\right]}
\\&\ge \exp\left(\alpha_{j^*}\cdot (t\val+\val_0) + \beta_{j^*}\right) \mathop{\mathbb E}\limits_{\mathrm{\mathbf{\hat{u}}}}\left[\exp\left(\gamma_{j^*}\cdot \mathbf{\hat{u}}\right)\right]\\
&=\exp\left(\alpha_{j^*}\cdot \val_0 + t(\alpha_{j^*}\cdot \val) + \beta_{j^*}\right) \mathop{\mathbb E}\limits_{\mathrm{\mathbf{\hat{u}}}}\left[\exp\left(\gamma_{j^*}\cdot \mathbf{\hat{u}}\right)\right]
\end{align*}
\end{small}

Since $\alpha_{j^*}\cdot \val>0$, by taking $t\to \infty$, the exponent would go to infinity, thus $LHS\to \infty$, and there exists $t_0$ such that $LHS>1$. Therefor, $\precond^\paracond(t\val+\val_0)$ is violated, deriving a contradiction.

\subsection{Proof of Theorem \ref{lem:opt}}
First, the objective function is convex.
We consider every constraint collected after Step 3. It is either: (1) a linear constraint, or (2) an instantiation $\precond^\paracond(\val^*)$ for some canonical constraint $\cancond(\guard,\precond^\paracond)$ and some $\val^*\in \guard$. For (2), we expand $\precond^\paracond(\val^*)$:
$$\sum_{j=1}^{k}{p_{j}\cdot \exp\left(\alpha_{j}\cdot \val^* + \beta_{j}\right)\cdot {\textstyle \mathop{\mathbb E}_{\valrd}}\left[\exp\left(\gamma_j\cdot \valrd\right)\right]}\le 1$$
By definition, to prove that $\Theta$ is a convex optimization problem,
it suffices to prove $LHS$ is a convex function over unknown coefficients.
First, since $\alpha_j,\beta_j$ are affine,
by convexity of $\exp(\cdot)$, $\exp\left(\alpha_{j}\cdot \val^* + \beta_{j}\right)$ is convex.
Then, we prove the expectation term ${\textstyle \mathop{\mathbb E}_{\valrd}}\left[\exp\left(\gamma_j\cdot \valrd\right)\right]$ is convex.
Since the composition of a convex function and a affine function is convex, and $\gamma_j$ is a affine function over unknown coefficients,
it suffices to prove that is ${\textstyle \mathop{\mathbb E}_{\valrd}}\left[\exp\left(\gamma\cdot \valrd\right)\right]$ convex over $\gamma$. We prove this by definition. Choose any pair $(\gamma,\gamma')$, and any $0\le c\le 1$:
\begin{small}
\begin{align*}
&{\textstyle \mathop{\mathbb E}_{\valrd}}\left[\exp\left((c\cdot \gamma+ (1-c)\cdot \gamma')\cdot \valrd\right)\right]
\\=\  &{\textstyle \mathop{\mathbb E}_{\valrd}}\left[\exp\left(c\cdot \gamma\cdot \valrd+ (1-c)\cdot \gamma'\cdot \valrd\right)\right]
\\\le\ & {\textstyle \mathop{\mathbb E}_{\valrd}}\left[c\cdot \exp\left(\gamma\cdot \valrd\right)+ (1-c)\cdot \exp\left( \gamma'\cdot \valrd\right)\right]
\\=\ &c\cdot {\textstyle \mathop{\mathbb E}_{\valrd}}\left[\exp\left(\gamma\cdot \valrd\right)\right]+(1-c)\cdot {\textstyle \mathop{\mathbb E}_{\valrd}}\left[\exp\left( \gamma'\cdot \valrd\right)\right]
\end{align*}
\end{small}
Since the product of two positive convex function is convex,
$\exp\left(\alpha_{j}\cdot \val^* + \beta_{j}\right)\cdot {\textstyle \mathop{\mathbb E}_{\valrd}}\left[\exp\left(\gamma_j\cdot \valrd\right)\right]$ is a convex function over unknown coefficients.
Finally, since $p_j\ge 0$, the whole $LHS$ is a non-negative combination of convex functions.
Therefore, the $LHS$ is a convex function over unknown coefficients.
Hence, $\Theta$ is a convex optimization problem.

  \section{Proofs of Section \ref{sec:alglower}}
\label{appendix:alglower}
\subsection{Proof of Theorem \ref{thm:soundlower}}
We first prove the strengthening in Step 4 is sound.
\begin{lemma}[Soundness of Strengthening]
\label{lem:relax}
If $\relaxpostcond$ is satisfied, then so is $\postcond^\paracond$.
\end{lemma}
\begin{proof}
By Theorem \ref{thm:jensen}, set the constant $Q:=\sum_{j=1}^k p_{j}.$ We have:
\begin{small}
\begin{align*}
&\sum_{j=1}^{k}{p_{j} \exp\left(\alpha_{j}\cdot \val + \beta_{j}\right) \mathop{\mathbb E}\limits_{\mathrm{\mathbf{\hat{u}}}}\left[\exp\left(\gamma_j\cdot \mathbf{\hat{u}}\right)\right]}\\
&\ge \sum_{j=1}^{k}{p_{j} \exp\left(\alpha_{j}\cdot \val + \beta_{j} + \gamma_j\cdot \mathop{\mathbb E}\limits_{\mathrm{\mathbf{\hat{u}}}}[\mathbf{\hat{u}}]\right) }\\
&= Q\cdot Q^{-1}\sum_{j=1}^{k}{p_{j} \exp\left(\alpha_{j}\cdot \val + \beta_{j} + \gamma_j\cdot \mathop{\mathbb E}\limits_{\mathrm{\mathbf{\hat{u}}}}[\mathbf{\hat{u}}]\right) }\\
&\ge Q\cdot \exp\left(Q^{-1}\sum_{j=1}^{k}{p_{j} (\alpha_{j}\cdot \val + \beta_{j} + \gamma_j\cdot \mathop{\mathbb E}\limits_{\mathrm{\mathbf{\hat{u}}}}[\mathbf{\hat{u}}])}\right) \\
&\ge Q\cdot \exp(-\ln Q) = 1
\end{align*}
\end{small}
\end{proof}

We are now ready to prove the theorem.
\begin{proof}[Proof of Theorem \ref{thm:soundlower}]
The constraints in Step 2 ensure boundness of $\lsol$.
By Lemma \ref{lem:relax}, we derive that if there exists a solution under strengthened constraints, then there exists a bounded post fixed-point $\lsol$.
The theorem follows.
\end{proof}

  \section{Evaluatoin Details}
\label{appendix:evaluation}
\subsection{Benchmarks}

Our benchmarks are presented below. Please also note that there are 3 example benchmarks in Section 3.

\begin{figure}[H]
  \lstset{language=prog}
  \lstset{tabsize=3}
  \begin{lstlisting}[mathescape]
  $i$:=$0$; $x$:=$0$;
  while($x\leq 99$)do
    switch do
      prob($0.5$):$\langle i,x\rangle$:=$\langle i+1,x+1\rangle$
      prob($0.5$):$\langle i,x\rangle$:=$\langle i+1,x\rangle$
    od;
      assert($x\le 200$)
   od
  \end{lstlisting}
  \caption{\textsc{RdAdder}: Randomized accumulation}
    \label{benchmark:rdadder}
\end{figure}
\begin{figure}[H]
  \lstset{language=prog}
  \lstset{tabsize=3}
  \begin{lstlisting}[mathescape]
  $i$:=$0$; $x$:=$0$; $ex$:=$0$; $cmd$:=$0$;
  while($i\leq 500$)do
    switch do
      prob($0.1$):$cmd:=1$ // SW
      prob($0.1$):$cmd:=2$ // SE
      prob($0.1$):$cmd:=3$ // W
      prob($0.1$):$cmd:=4$ // E
            $\cdots$
      prob($0.2$):$cmd:=9$ // STAY
    od;
        if ($cmd$ == $1$) then // SW Action
            switch do  //Add noise
                prob($0.5$):
                    $i:=i+1$
                    $x:=x-1.414-0.05$
                    $ex:=ex-1.414$
                prob($0.5$):
                    $i:=i+1$
                    $x:=x-1.414+0.05$
                    $ex:=ex-1.414$
            od
        else
        $\cdots$
        if ($cmd$ == $9$) then
            $\cdots$
        fi

   od
   assert($x-ex\ge -3$)
  \end{lstlisting}
  \caption{\textsc{Robot}: Deadrock robot}
    \label{benchmark:robot}
\end{figure}
\begin{figure}[H]
  \lstset{language=prog}
  \lstset{tabsize=3}
  \begin{lstlisting}[mathescape]
  $x$:=$0$;
  while($x\ge 0$)do
      assert($x\le 1000$)
    switch do
      prob($0.5$):$\langle x\rangle$:=$\langle x-2\rangle$
      prob($0.5$):$\langle x\rangle$:=$\langle x+1\rangle$
    od;
   od
  \end{lstlisting}
  \caption{\textsc{1DWalk}: 1D random walk with assertions}
    \label{benchmark:1dassert}
\end{figure}
\begin{figure}[H]
  \lstset{language=prog}
  \lstset{tabsize=3}
  \begin{lstlisting}[mathescape]
  $x$:=$0$; $y$:=$0$;
  while($y\ge 1$)do
      if prob($0.5$) then
    switch do
      prob($0.75$):$x$:=$x+1$
      prob($0.25$):$x$:=$x-1$
    od;
      else
    switch do
      prob($0.75$):$y$:=$y-1$
      prob($0.25$):$y$:=$y+1$
    od;
      fi
      assert($x\ge 1$)
   od
  \end{lstlisting}
  \caption{\textsc{2DWalk}: 2D random walk with assertions}
    \label{benchmark:2dassert}
\end{figure}
\begin{figure}[H]
  \lstset{language=prog}
  \lstset{tabsize=3}
  \begin{lstlisting}[mathescape]
  $x$:=$0$; $y$:=$0$;
  while($x\ge 0\land y\ge 0\land z\ge 0$)do
      assert($x+y+z\le 1000$)
      if prob($0.9$) then
        switch do
        prob($0.5$):$\langle x,y \rangle$:=$\langle x-1,y-1 \rangle$
        prob($0.5$):$\langle z \rangle$:=$\langle z-1 \rangle$
      od;
      else
      switch do
        prob($0.5$):$\langle x,y \rangle$:=$\langle x+0.1,y+0.1 \rangle$
        prob($0.5$):$\langle z \rangle$:=$\langle z+0.1 \rangle$
      od;
      fi
   od
  \end{lstlisting}
  \caption{\textsc{3DWalk}: 3D random walk with assertions}
    \label{benchmark:3dassert}
\end{figure}

\begin{figure}[H]
  \lstset{language=prog}
  \lstset{tabsize=3}
  \begin{lstlisting}[mathescape]
  $i$:=$0$; $t$:=$0$;
  while($i\le 5$)do
      if $i=0$ then
        $i:=i+1$
        $t:=t+1$
      else
      if $i=1$ then
        if prob($0.8$) then
          $i:=i+1$
          $t:=t+1$
        else
          $t:=t+1$
        fi
      else
      $\cdots$
      if $i=4$ then
        if prob($0.2$) then
          $i:=i+1$
          $t:=t+1$
        else
          $t:=t+1$
        fi
      fi
      assert($t\le 100$)
   od
  \end{lstlisting}
  \caption{\textsc{Coupon}: Concentration of running time of coupon collector with $5$ items}
    \label{benchmark:coupon}
\end{figure}

\begin{figure}[H]
  \lstset{language=prog}
  \lstset{tabsize=3}
  \begin{lstlisting}[mathescape]
  $x$:=$0$; $y$:=$0$; $t$:=$0$;
  while($x+3\le 50$)do
      if($y\le 49$) then
        if prob(0.5) then
          $\langle y,t\rangle$:=$\langle y+1,t+1\rangle$
        else
          $\langle y,t\rangle$:=$\langle y,t+1\rangle$
        fi
      else
        switch do
          prob($0.25$): $\langle x,t\rangle$:=$\langle x,t+1\rangle$
          prob($0.25$): $\langle x,t\rangle$:=$\langle x+1,t+1\rangle$
          prob($0.25$): $\langle x,t\rangle$:=$\langle x+2,t+1\rangle$
          prob($0.25$): $\langle x,t\rangle$:=$\langle x+3,t+1\rangle$
        fi
      fi
      assert($t\le 100$)
   od
  \end{lstlisting}
  \caption{\textsc{Prspeed}: Concentration of running time of random walk with randomized speed}
    \label{benchmark:prspeed}
\end{figure}

\begin{figure}[H]
  \lstset{language=prog}
  \lstset{tabsize=3}
  \begin{lstlisting}[mathescape]
  $i$:=$0$; $p$:=$10^{-7}$
    while ($i\le 40$) do
      if prob($(1-p)^5$) then
        skip //ABSTRACTED
      else
        exit
      fi
      if prob($0.9999$) then
        skip //ABSTRACTED
      else
        exit
      fi
      if prob($0.9999$) then
        skip //ABSTRACTED
      else
        exit
      fi
      if prob($(1-p)^3$) then
        skip //ABSTRACTED
      else
        exit
      fi
      if prob($(1-p)^6$) then
        skip //ABSTRACTED
      else
        exit
      fi
      $i$:=$i+1$
    od

  \end{lstlisting}
  \caption{\textsc{Newton}: Executing Newton's iteration algorithm on unreliable hardware (Abstracted version)}
    \label{benchmark:hardwarenewton}
\end{figure}

\begin{figure}[H]
  \lstset{language=prog}
  \lstset{tabsize=3}
  \begin{lstlisting}[mathescape]
  $i$:=$0$; $p$:=$10^{-7}$
    while ($i\le 19$) do
      $j$:=$0$
      while ($j\le 15$) do
        $k$:=$0$
        while ($k\le 15$) do
          if prob($(1-p)^3$) then
            skip    //ABSTRACTED
          else
            exit
          fi
          $k$:=$k+1$
        od
        $j$:=$j+1$
      od
      if prob($1-p$) then
        skip //ABSTRACTED
      else
        exit
      fi
      $i$:=$i+1$
    od

  \end{lstlisting}
  \caption{\textsc{Ref}: Executing Searchref algorithm on unreliable hardware (Abstracted version)}
    \label{benchmark:hardwareref}
\end{figure}

\subsection{Detailed result}
The detailed result is listed in Table \ref{table:expsym1} Table \ref{table:expsym2} Table \ref{table:expsym3}, where we report the symbolic bound for every benchmark.

  \renewcommand{\arraystretch}{1.2}
\begin{table*}
	\vspace{-.5em}
	\begin{footnotesize}
		\begin{tabular}{|c|c|c||c|}
			\hline
			\multicolumn{2}{|c|}{\textbf{Benchmark}} & \textbf{Parameters} & \textbf{Algorithm of Section~\ref{sec:heuristicalg}}  \\
			\hline
			\hline
			\multirow{6}{*}{\begin{turn}{90}\textsc{Deviation}\end{turn}}
			& \multirow{3}{*}{\textsc{RdAdder}}
			& {$\Pr[X-\mathbb E[X]\ge 25]$} & $\exp(8\cdot0.05\cdot(-1.0\cdot x+ 0.45 \cdot i -25.25))$    \\
			\cline{3-4}
			& & {$\Pr[X-\mathbb E[X]\ge 50]$} & $\exp(8\cdot0.02\cdot(-1.0\cdot x + 0.47 \cdot i -12.75))$   \\
			\cline{3-4}
			& & {$\Pr[X-\mathbb E[X]\ge 75]$} &$\exp(8\cdot0.07\cdot(-1.0\cdot x + 0.42 \cdot i -37.75))$    \\
			\cline{2-4}
			& \multirow{3}{*}{\textsc{Robot}}
			& {$\Pr[X-\mathbb E[X]\ge 1.8]$} &  $\exp(8\cdot0.07\cdot(-0.14\cdot i -10 \cdot x + 10 \cdot ex + 0\cdot  dxc-9))$   \\
			\cline{3-4}
			& & {$\Pr[X-\mathbb E[X]\ge 2.0]$} &  $\exp(8\cdot0.08\cdot(-0.16\cdot i -10 \cdot x + 10 \cdot ex + 0\cdot  dxc-10))$    \\
			\cline{3-4}
			& & {$\Pr[X-\mathbb E[X]\ge 2.2]$} &  $\exp(8\cdot0.09\cdot(-0.18\cdot i -10 \cdot x + 10 \cdot ex + 0\cdot  dxc-11))$   \\
			\hline

			\multirow{9}{*}{\begin{turn}{90}\textsc{Concentration}\end{turn}}
			& \multirow{3}{*}{\textsc{Coupon}}
			& $\Pr[T>100]$ & $\exp(8\cdot0.03(-1\cdot i + 0.12 \cdot t -7.60))$   \\
			\cline{3-4}
			& & $\Pr[T>300]$ &  $\exp(8\cdot0.04\cdot(-1\cdot i + 0.10 \cdot t -27.60))$    \\
			\cline{3-4}
			& & $\Pr[T>500]$ &  $\exp(8\cdot0.04\cdot(-1\cdot i + 0.10 \cdot t -47.57))$    \\
			\cline{2-4}
			& \multirow{3}{*}{\textsc{Prspeed}}
			& $\Pr[T>150]$  &$\exp(8\cdot0.06\cdot(-0.33\cdot x -1 \cdot y + 0.29 \cdot t -32.75))$     \\
			\cline{3-4}
			& & $\Pr[T>200]$ & $\exp(8\cdot0.07\cdot(-0.33\cdot x -1 \cdot y + 0.28 \cdot t -45.24))$  \\
			\cline{3-4}
			& & $\Pr[T>250]$ & $\exp(8\cdot0.06\cdot(-0.33\cdot x -1 \cdot y + 0.31 \cdot t -20.24))$   \\
			\cline{2-4}
			& \multirow{3}{*}{\textsc{Rdwalk}}
			& $\Pr[T>400]$ & $\exp(8\cdot0.03\cdot(-0.5\cdot x + 0.17\cdot y -37.62))$  \\
			\cline{3-4}
			& & $\Pr[T>500]$ &  $\exp(8\cdot0.03\cdot(-0.5\cdot x + 0.18\cdot y -25.12))$ \\
			\cline{3-4}
			& & $\Pr[T>600]$ &  $\exp(8\cdot0.04\cdot(-0.5\cdot x + 0.16\cdot y -40.12))$ \\
			\hline

			\multirow{12}{*}{\begin{turn}{90}\textsc{StoInv}\end{turn}}
			& \multirow{3}{*}{\textsc{1DWalk}}
			& $x=10$ &$\exp(8\cdot0.05 \cdot(0.33 \cdot x -333.55))$    \\
			\cline{3-4}
			& & $x=50$ &  $\exp(8\cdot 0.05 \cdot(0.33 \cdot x -333.55))$   \\
			\cline{3-4}
			& & $x=100$ & $\exp(8\cdot 0.05 \cdot(0.33 \cdot x -333.55))$   \\
			\cline{2-4}
			& \multirow{3}{*}{\textsc{2DWalk}}
			& $(x,y)=(1000,10)$ & $\exp(8\cdot0.04\cdot(-0.5\cdot x + 0 \cdot y +7.99\cdot 10^{-8}))$   \\
			\cline{3-4}
			& & $(x,y)=(500,40)$  & $\exp(8\cdot 0.04 \cdot(-0.5\cdot x + 0 \cdot y +1.8\cdot 10^{-7}))$   \\
			\cline{3-4}
			& & $(x,y)=(400,50)$ &$\exp(8\cdot 0.04 \cdot(-0.5\cdot x + 0 \cdot y +1.8\cdot 10^{-7}))$  \\
			\cline{2-4}
			& \multirow{3}{*}{\textsc{3DWalk}}
			& $\!(x,\!y,\!z)\!=\!(100,\!100,\!100)\!$ & $\exp(8\cdot0.19\cdot(0.58\cdot x + 0.58 \cdot y + 0.58\cdot z -487.80))$   \\
			\cline{3-4}
			& & $\!(x,\!y,\!z)\!=\!(100,\!150,\!200)\!$ & $\exp(8\cdot 0.19\cdot(0.58\cdot x + 0.58 \cdot y + 0.58\cdot z -487.90))$   \\
			\cline{3-4}
			& & $\!(x,\!y,\!z)\!=\!(300,\!100,\!150)\!$ &  $\exp(8\cdot 0.19\cdot(0.58\cdot x + 0.58 \cdot y + 0.58\cdot z -487.80))$   \\
			\cline{2-4}
			& \multirow{3}{*}{\textsc{Race}} & $(x,y)=(40,0)$ &  $\exp(8\cdot0.08\cdot(-0.67\cdot x + 0.5\cdot y + 16.58)$    \\
			\cline{3-4}
			& & $(x,y)=(35,0)$ &  $\exp(8\cdot0.07\cdot(-0.63\cdot x + 0.5\cdot y + 13.34)$     \\
			\cline{3-4}
			& & $(x,y)=(45,0)$ &  $\exp(8\cdot0.10\cdot(-0.70\cdot x + 0.5\cdot y + 20.41)$     \\
			\hline

		\end{tabular}
	\end{footnotesize}
	\caption{Symbolic Results for Upper-bound Benchmarks of Algorithm of Section~\ref{sec:heuristicalg} .}
	\label{table:expsym1}
\end{table*}
\begin{table*}
	\vspace{-.5em}
	\begin{footnotesize}
		\begin{tabular}{|c|c|c||c|}
			\hline
			\multicolumn{2}{|c|}{\textbf{Benchmark}} & \textbf{Parameters} & \textbf{Algorithm of Section~\ref{sec:soundcompalg}} \\
			\hline
			\hline
			\multirow{6}{*}{\begin{turn}{90}\textsc{Deviation}\end{turn}}
			& \multirow{3}{*}{\textsc{RdAdder}}
			& {$\Pr[X-\mathbb E[X]\ge 25]$}   & $\exp(-0.20\cdot x + 0.09 \cdot i -2.6)$    \\
			\cline{3-4}
			& & {$\Pr[X-\mathbb E[X]\ge 50]$}   & $\exp(-0.40 \cdot x + 0.18 \cdot i-10.25)$  \\
			\cline{3-4}
			& & {$\Pr[X-\mathbb E[X]\ge 75]$}   & $\exp(-0.62 \cdot x + 0.26 \cdot i -23.11)$     \\
			\cline{2-4}
			& \multirow{3}{*}{\textsc{Robot}}
			& {$\Pr[X-\mathbb E[X]\ge 1.8]$}  &  $\exp(-0.22\cdot i -13.85 \cdot x + 13.85 \cdot ex + 0\cdot  dxc-11.55)$     \\
			\cline{3-4}
			& & {$\Pr[X-\mathbb E[X]\ge 2.0]$}  & $\exp(-0.29\cdot i -16.09 \cdot x + 16.09 \cdot ex + 0\cdot  dxc-14.55)$    \\
			\cline{3-4}
			& & {$\Pr[X-\mathbb E[X]\ge 2.2]$}  &  $\exp(-0.38\cdot i -18.70 \cdot x + 18.70 \cdot ex + 0\cdot  dxc-18.00)$   \\
			\hline

			\multirow{9}{*}{\begin{turn}{90}\textsc{Concentration}\end{turn}}
			& \multirow{3}{*}{\textsc{Coupon}}
			& $\Pr[T>100]$ &  $\exp(-1.56\cdot i + 0.17 \cdot t -9.56)$  \\
			\cline{3-4}
			& & $\Pr[T>300]$  & $\exp(-2.69\cdot i + 0.20 \cdot t -48.65)$     \\
			\cline{3-4}
			& & $\Pr[T>500]$  & $\exp(-3.21\cdot i + 0.21 \cdot t -90.71)$     \\
			\cline{2-4}
			& \multirow{3}{*}{\textsc{Prspeed}}
			& $\Pr[T>150]$   &  $\exp(-0.51\cdot x -2.45 \cdot y + 0.61 \cdot t -63.39)$     \\
			\cline{3-4}
			& & $\Pr[T>200]$  & $\exp(-0.53\cdot x -2.70 \cdot y +0.62 \cdot t -92.96)$   \\
			\cline{3-4}
			& & $\Pr[T>250]$  & $\exp(-0.47\cdot x -2.12 \cdot y + 0.58 \cdot t -35.54)$    \\
			\cline{2-4}
			& \multirow{3}{*}{\textsc{Rdwalk}}
			& $\Pr[T>400]$  &  $\exp(-0.34\cdot x + 0.12\cdot y -27.18)$   \\
			\cline{3-4}
			& & $\Pr[T>500]$ &  $\exp(-0.29\cdot x + 0.11\cdot y -15.35)$   \\
			\cline{3-4}
			& & $\Pr[T>600]$  &  $\exp(-0.38\cdot x + 0.12\cdot y -39.87)$   \\
			\hline

			\multirow{12}{*}{\begin{turn}{90}\textsc{StoInv}\end{turn}}
			& \multirow{3}{*}{\textsc{1DWalk}}
			& $x=10$   &  $\exp(0.48 \cdot x -481.69)$   \\
			\cline{3-4}
			& & $x=50$  &  $\exp(0.48 \cdot x -481.69)$    \\
			\cline{3-4}
			& & $x=100$  &  $\exp(0.48\cdot -481.69)$   \\
			\cline{2-4}
			& \multirow{3}{*}{\textsc{2DWalk}}
			& $(x,y)=(1000,10)$  & $\exp(-1.31\cdot x + 0.54 \cdot y -3.02\cdot 10^{-9})$    \\
			\cline{3-4}
			& & $(x,y)=(500,40)$   &  $\exp(-1.31 \cdot x + 0.48 \cdot y -1.46\cdot 10^{-9})$   \\
			\cline{3-4}
			& & $(x,y)=(400,50)$  &  $\exp(-1.31 \cdot x + 0.44 \cdot y -2.44\cdot 10^{-9})$ \\
			\cline{2-4}
			& \multirow{3}{*}{\textsc{3DWalk}}
			& $\!(x,\!y,\!z)\!=\!(100,\!100,\!100)\!$  & $\exp(9.22\cdot x + 9.22 \cdot y + 9.22\cdot z -9.22\cdot 10^{3})$     \\
			\cline{3-4}
			& & $\!(x,\!y,\!z)\!=\!(100,\!150,\!200)\!$  &  $\exp(9.22\cdot x + 9.22 \cdot y + 9.22 \cdot z -9.22\cdot 10^{3})$    \\
			\cline{3-4}
			& & $\!(x,\!y,\!z)\!=\!(300,\!100,\!150)\!$  & $\exp(9.22\cdot x + 9.22\cdot y + 9.22\cdot z -9.22\cdot 10^{3})$    \\
			\cline{2-4}
			& \multirow{3}{*}{\textsc{Race}} & $(x,y)=(40,0)$   &  $\exp(-1.18\cdot x + 0.85\cdot y + 31.79)$  \\
			\cline{3-4}
			& & $(x,y)=(35,0)$   &  $\exp(-0.82 \cdot x+ 0.63\cdot y + 18.19)$    \\
			\cline{3-4}
			& & $(x,y)=(45,0)$   &  $\exp(-0.82\cdot x + 0.63\cdot y + 18.19)$   \\
			\hline

		\end{tabular}
	\end{footnotesize}
	\caption{Symbolic Results for Upper-bound Benchmarks of Algorithm of Section~\ref{sec:soundcompalg}.}
	\label{table:expsym2}
\end{table*}
\begin{table*}
	\vspace{-1.7em}
	\begin{footnotesize}
		\begin{tabular}{|c|c|c||c|}
			\hline
			\multicolumn{2}{|c|}{{\textbf{Benchmark}}} & \multicolumn{1}{c||}{{\textbf{Parameters}}} & {\textbf{Algorithm of Section~\ref{sec:alglower}}}                         \\ \hline \hline
			\multirow{9}{*}{\begin{turn}{90}\textsc{Hardware}\end{turn}}
			& \multirow{3}{*}{\textsc{M1DWalk}} & $p=10^{-7}$  & $\exp(2\cdot 10^{-7} \cdot x -2\cdot 10^{-4})$ \\
			\cline{3-4}
			& & $p=10^{-5}$  & $\exp(2\cdot 10^{-4} \cdot x -0.002)$ \\
			\cline{3-4}
			& & $p=10^{-4}$  & $\exp(2\cdot 10^{-4} \cdot x -0.02)$ \\
			\cline{2-4}
			& \multirow{3}{*}{\textsc{Newton}} & $p=5\cdot 10^{-4}$  & $\exp(7.7\cdot 10^{-3} \cdot i -0.31)$  \\
			\cline{3-4}
			& & $p=10^{-3}$ & $\exp(1.52\cdot 10^{-2} \cdot i -0.62)$  \\
			\cline{3-4}
			& & $p=1.5\cdot 10^{-3}$ & $\exp(2.27\cdot 10^{-2}\cdot i -0.93)$ \\
			\cline{2-4}
			& \multirow{3}{*}{\textsc{Ref}}
			& $p=10^{-7}$ & $\exp(7.69\cdot 10^{-4} \cdot i + 0 \cdot j + 0 \cdot k -0.015)$  \\
			\cline{3-4}
			& & $p=10^{-6}$ & $\exp(7.7\cdot 10^{-3} \cdot i + 0 \cdot j + 0 \cdot k -0.15)$  \\
			\cline{3-4}
			& & $p=10^{-5}$ & $\exp(7.7\cdot 10^{-2} \cdot i + 0 \cdot j + 0 \cdot k -1.53)$   \\
			\hline
		\end{tabular}
	\end{footnotesize}
	\caption{Symbolic Results for Lower-bound Benchmarks of Algorithm of Section~\ref{sec:alglower}.}
	\vspace{-1.5em}
	\label{table:expsym3}
\end{table*}
\renewcommand{\arraystretch}{1}

\end{document}